\RequirePackage{fix-cm}
\documentclass[smallextended,envcountsame]{svjour3}       
\smartqed  
\usepackage{graphicx}
\graphicspath{ {./Images/}}
\usepackage[T1]{fontenc}
\usepackage{lmodern}
\usepackage{multirow}
\usepackage{amsmath,amsfonts}
\usepackage{epsfig}
\usepackage{comment}
\usepackage[dvipsnames]{xcolor}
\usepackage{todonotes}
\newtheorem{observation}{Observation}

\usepackage{color}
\usepackage[pagewise, displaymath, mathlines]{lineno}
\linenumbers

\newcommand{\blue}[1]{{\textcolor{black}{#1}}}
\newcommand{\green}[1]{{\textcolor{black}{#1}}}

\newcommand{\remove}[1]{}
\newcommand{\IR}{\mathbb{R}}
\newcommand{\Pb}[3]{
	\smallskip
	\noindent\fbox{
		\begin{minipage}{.95\textwidth}
			\begin{tabular*}{\textwidth}{@{\extracolsep{\fill}}lr} Problem: #1 \\ \end{tabular*}
			{\bf{Input:}} #2  \\
			{\bf{Output:}} #3
		\end{minipage}
	}
	\medskip}

\begin{document}
\nolinenumbers
\title{Complexity and Approximation for Discriminating and Identifying Code Problems in Geometric Setups\footnote{A preliminary version of this paper was presented at the ISAAC 2020 conference and appeared as~\cite{isaac}.}
}

\author{Sanjana Dey         \and
        Florent Foucaud 	\and
        Subhas C. Nandy		\and
        Arunabha Sen		
}

\authorrunning{S. Dey et al.}
\titlerunning{Complexity and Approximation for Discriminating and Identifying Code Problems in Geometric Setups}
\institute{Sanjana Dey \at
	ACM Unit, Indian Statistical Institute, Kolkata, India\\
	\email{info4.sanjana@gmail.com}           
	\and
	Florent Foucaud \at
        LIMOS, CNRS UMR 6158, Universit\'e Clermont Auvergne, Aubi\`ere, France\\
	Univ. Bordeaux, Bordeaux INP, CNRS, LaBRI, UMR5800, F-33400 Talence, France\\
        Univ. Orl\'eans, INSA Centre Val de Loire, LIFO EA 4022, F-45067 Orl\'eans, France\\
	This author was partially funded by the French government IDEX-ISITE initiative 16-IDEX-0001 (CAP 20-25), by the ANR project GRALMECO (ANR-21-CE48-0004-01), by the ANR project HOSIGRA (ANR-17-CE40-0022) and by the IFCAM project ''Applications of graph homomorphisms'' (MA/IFCAM/18/39).\\
	\email{florent.foucaud@uca.fr}
	\and
	Subhas C. Nandy \at
	ACM Unit, Indian Statistical Institute, Kolkata, India\\
	\email{subhas.c.nandy@gmail.com}
	\and
	Arunabha Sen \at
	Arizona State University, Tempe, Arizona 85287, USA\\
	\email{arunabha.sen@asu.edu} 
}

\date{}

\maketitle

\begin{abstract}
We study geometric variations of the discriminating code problem. In the \emph{discrete version} of the problem, a finite set of points $P$ and a finite set of objects $S$ are given in $\mathbb{R}^d$. The objective is to choose a subset $S^* \subseteq S$ of minimum cardinality such that for each point $p_i \in P$, the subset $S_i^* \subseteq S^*$ covering $p_i$ satisfies $S_i^*\neq \emptyset$, and each pair $p_i,p_j \in P$, $i \neq j$, we have $S_i^* \neq S_j^*$. In the \emph{continuous version} of the problem, the solution set $S^*$ can be chosen freely among a (potentially infinite) class of allowed geometric objects.

In the 1-dimensional case ($d=1$), the points in $P$ are placed on a horizontal line $L$, and the objects in $S$ are finite-length line segments aligned with $L$ (called intervals). We show that the discrete version of this problem is NP-complete. This is somewhat surprising as the continuous version is known to be polynomial-time solvable. This is also in contrast with most geometric covering problems, which are usually polynomial-time solvable in one dimension. Still for the 1-dimensional discrete version, we design a polynomial-time $2$-approximation algorithm. We also design a PTAS for both discrete and continuous versions in one dimension, for the restriction where the intervals are all required to have the same length.

We then study the 2-dimensional case ($d=2$) for axis-parallel unit square objects. We show that both continuous and discrete versions are NP-complete, and design polynomial-time approximation algorithms that produce $(16\cdot OPT+1)$-approximate and $(64\cdot OPT+1)$-approximate solutions respectively, using rounding of suitably defined integer linear programming problems.

Finally, we apply our techniques to a related variant of the discrete problem, where instead of points and geometric objects we just have a set $S$ of objects. The goal is to select a small subset $S^*$ of objects so that all objects of $S$ are discriminated by their intersection with the objects of $S^*$. This problem can be viewed as a graph problem by stating it in terms of the vertices of the geometric intersection graph of $S$. Under this graph-theoretical form, it is known as the \emph{identifying code problem}. 
We show that the identifying code problem for axis-parallel unit square intersection graphs (in $d=2$) can be solved in the same manner as for the discrete version of the discriminating code problem for unit square objects described above, and all our positive approximation results still hold in this setting.
\end{abstract}
 

\section{Introduction} \label{intro}
We consider geometric versions of the \textsc{Discriminating Code} problem, which are variations of classical geometric covering problems. Here, a set of point sites $P=\{p_1,p_2, \ldots, p_n\}$ is given in $\mathbb{R}^d$. Let $S$ be a set of geometric objects (i.e. closed curves along with their interior) in $\mathbb{R}^d$, where each $s_i \in S$ has a unique identification. We use $S_i\subseteq S$ to denote the set of objects containing $p_i \in P$. The objective is to choose a minimum-size subset $S^* \subseteq S$ of objects such that (i) the subset $S^*_i \subset S^*$ containing $p_i$ is not an empty set for each $i=1,2,\ldots, n$, and (ii) each pair of points $p_i, p_j \in P$, $i\neq j$, satisfy $S^*_i \neq S^*_j$. Thus, (i) suggests that each $p_i$ has an identification, and (ii) suggests that $p_i$ and $p_j$ ($i \neq j$) are discriminated by at least one element in $S^*_i$ or $S^*_j$. From now onwards, we will use $d$ to denote $d$-dimension for some integer $d$.  

In the \emph{discrete} version, both $P$ and $S$ are given as input. In the \emph{continuous} version, only the points in $P$ are given, and the objects can be chosen freely (among some infinite class of allowed objects).


The problem is motivated as follows. Consider a terrain that is difficult to navigate. A set of sensors, each assigned a unique identification number ($id$), are deployed in that terrain, all of which  can communicate with a single base station. If a region of the terrain suffers from some specific problem, a subset of sensors will detect that and inform the base station. From the $id$'s of the alerted sensors, one can uniquely identify the affected region, and a rescue team can be sent. The covering zone of each sensor can be represented by an object in $S$. In the \emph{arrangement} of these objects (that is, the partition of the plane into \emph{cells} defined by the union of boundaries of all objects, each cell corresponds to a face of this union~\cite{Berg}) divides the entire plane into regions. A representative point of each region may be considered as a site. The set $P$ consists of some of those sites. We need to determine the minimum number of sensors such that no two sites in $P$ are covered by the same set of $id$s. Apart from coverage problems in sensor networks, this problem has applications in fault detection, heat prone zone in VLSI circuits, disaster management, environmental monitoring, localization and contamination detection~\cite{Laifenfeld,Ray2004RobustLD}, to name a few.

The general version of the problem has been formulated as a combinatorial problem in a bipartite graph in~\cite{CharbitCCH06,CharonCHL08} as follows.

\Pb{\textsc{Minimum Discriminating Code} (\textsc{Min-Disc-Code})}{A connected bipartite graph $G=(U\cup V,E)$, where $E \subseteq \{(u,v)|u \in U, v \in V\}$.}{A minimum-size subset $U^* \subseteq U$ such that $U^* \cap N(v) \neq \emptyset$ for all $v \in V$, and $U^*\cap N(v) \neq U^*\cap N(v')$ for every pair $v, v' \in V$, $v\neq v'$.}

In the geometric version of the \textsc{Min-Disc-Code} problem introduced in~\cite{DRCN}, which will be further referred to as \textsc{G-Min-Disc-Code}, the two sets of nodes in the bipartite graph $G$ are $U=S$, a set of geometric objects, and $V=P$, a set of points in $\mathbb{R}^d$; an object in $S$ is adjacent to all the points in $P$ that it contains. \green{(Graph $G$ is called the \emph{incidence graph of $(P,S)$.})} The {\em id} of a point $p \in P$ with respect to the set $S$ is the union of the id's of the subset $S'\subseteq S$ that contains $p$. Given an instance $(P,S)$, two points $p_i,p_j \in P$ are called {\em twins} if each member in $S$ that contains $p_i$ also contains $p_j$, and vice-versa. An instance $(P,S)$ of \textsc{G-Min-Disc-Code} is {\em twin-free} if no two points in $P$ are twins. 
As mentioned earlier, for a twin-free instance, a subset of $S$ that can uniquely assign  id's to all the points in $P$ is said to \emph{discriminate} the points of $P$ and is called a {\em discriminating code}. 
In the discrete version of the problem, the set $S$ of objects is also given along with the set of points $P$ as the input; the objective is to find a subset $S^* \subseteq S$ of minimum cardinality that is a discriminating code for the points in $P$. In the continuous version, we can freely choose the objects $S^*$ in $\mathbb{R}^d$ such that each point gets a unique id, and the size of the set $S^*$ is minimum. The two problems are formally stated as follows.

\Pb{\textsc{Continuous-G-Min-Disc-Code (for objects of a prescribed type)}}
{A point set $P$ to be discriminated.}
{A minimum-size set $S^*$ of objects of the prescribed type, placed {\em anywhere} in the region under consideration, that discriminates the points in $P$.}

\Pb{\textsc{Discrete-G-Min-Disc-Code}}
{A point set $P$ to be discriminated, and a set of objects $S$ to be used for the discrimination.}
{A minimum-size subset $S^* \subseteq S$ that discriminates all points in $P$.}

\green{We can use~\cite{twins} to show the following.}

\begin{proposition}\label{twin-free}
Checking whether a given instance $(P,S)$ of \textsc{Discrete-G-Min-Disc-Code} with $|P|=n$ and $|S|=m$ is twin-free can be done in time $O(m\cdot n)$.
\end{proposition}
\begin{proof}
  \green{It is known that checking whether a graph has twins (vertices with the same neighbourhood) can be done in linear time (in terms of the number of vertices and edges) using the technique of \emph{partition refinement} (see Algorithm~2 from~\cite[Section~3]{twins}). 
    As mentioned before, our geometric instance $(P,S)$ can be associated to its incidence graph $G=(U\cup V,E)$ and is twin-free if and only if there is no pair of graph-twins inside $V$ in $G$. Thus, we can build the graph $G$ in linear time by directly using the description of $(P,S)$, and then applying the linear-time algorithm from~\cite{twins}. The graph $G$ has $m+n$ vertices and $O(m\cdot n)$ edges and thus, the running time follows.}\hfill\qed
\end{proof}

Note that one can also use the above method to decide whether an instance $P$ of \textsc{Continuous-G-Min-Disc-Code} is twin-free, by computing all feasible equivalence classes of objects that can be placed (with respect to the intersections with $P$). However, in general there can be as many as $2^{|P|}$ such equivalence classes, and it might not always be trivial to compute them: this would depend on the context.

A well-studied special case of \textsc{Min-Disc-Code} for graphs is the problem \textsc{Minimum Identifying Code}, initially defined in~\cite{KarpovskyCL98} as follows (where $N[v_i]$ denotes the set of vertices adjacent to $v_i$, together with $v_i$ itself):

\Pb{\textsc{Minimum Identifying Code} (\textsc{Min-ID-Code})}{A graph $G=(V,E)$.}{A subset $V^* \subseteq V$ of minimum size such that $V^*$ is a dominating set, and $V^*\cap N[v_i] \neq V^*\cap N[v_j]$ for every pair $v_i, v_j \in V$, $i\neq j$.}

The \textsc{Min-ID-Code} problem is a special case of \textsc{Min-Disc-Code}~\cite{CharbitCCH06,Foucaud15}: indeed, given a graph $G$, let $B(G)$ denote the closed neighbourhood incidence bipartite graph of $G$ (one part consists of the vertex set of $G$, and the other part, of the (multi-)set of all closed neighbouroods of vertices of $G$; a closed neighbourhood vertex is adjacent to all the vertices it contains). Then, \textsc{Min-ID-Code} for $G$ is equivalent to \textsc{Min-Disc-Code} for $B(G)$. Both \textsc{Min-ID-Code} and \textsc{Min-Disc-Code} are NP-complete~\cite{CharbitCCH06,CharonCHL08,CharonHL03}. A polynomial-time algorithm is available for \textsc{Min-Disc-Code} on trees~\cite{CharonCHL08}. It was shown that \textsc{Min-ID-Code} for graphs (and thus \textsc{Min-Disc-Code}) is $\log$-APX hard~\cite{LaifenfeldT08}, and this holds even for split graphs, bipartite graphs, co-bipartite graphs~\cite{Foucaud15}, and for bipartite graphs of girth~6~\cite{BousquetLLPT15}. However, for line graphs and planar graphs, \textsc{Min-ID-Code} remains NP-complete but constant factor approximation algorithms are available (see~\cite{FoucaudGNPV13} and~\cite{BazganFS19}, respectively). 

\textsc{Min-ID-Code} was studied in particular for the related setting of geometric intersection graphs, for example on unit disk graphs~\cite{MS09} and interval graphs~\cite{BousquetLLPT15,Foucaud,FoucaudMNPV17}. The problem remains NP-complete for these graph classes, and a 6-approximation algorithm exists for interval graphs~\cite{BousquetLLPT15}. For unit interval graphs, the computational hardness is not known, but a polynomial-time approximation scheme (PTAS) algorithm is given in the second author's PhD thesis~\cite{Foucaud}. A PTAS also exists for \textsc{Min-ID-Code} on planar graphs~\cite{BazganFS19}.

\noindent{\bf Further related work.}
In the context of the above-mentioned practical applications, the \textsc{Discrete-G-Min-Disc-Code} problem in 2D was defined in~\cite{DRCN}, where an integer programming formulation (ILP) of the problem was given along with an experimental study. The \textsc{Continuous-G-Min-Disc-Code} problem was introduced under a different name in~\cite{GledelP19}, and shown to be NP-complete for disks in 2D, but polynomial-time in 1D (even when the intervals are restricted to have bounded length). These two problems are related to the class of \emph{geometric covering problems}, for which also both the discrete and continuous version are studied extensively~\cite{covering}. A related problem is the \textsc{Test Cover} problem~\cite{BontridderHHHLRS03}, which is similar to \textsc{Min-Disc-Code} (but defined on hypergraphs). It is equivalent to the variant of \textsc{Min-Disc-Code} where the covering condition ``$U^* \cap N(v) \neq \emptyset$'' is not required. Thus, a discriminating code is a test cover, but the converse may not be true, since there may exist a vertex uncovered by a test cover. Geometric versions of \textsc{Test Cover} have been studied under various names. For example, the \emph{separation} problems in~\cite{Boland1995,CDKW05,HM20} can be seen as continuous geometric versions of \textsc{Test Cover} in 2D, where the objects are half-planes. Similar problems are also sometimes called \emph{shattering} problems, see~\cite{NANDY02}.

More references on several coding mechanisms on graphs based on different applications, namely locating-dominating sets, open locating dominating sets, metric dimension, etc, and their computational hardness results are available in~\cite{Foucaud,FoucaudMNPV17}.

\noindent{\bf Our results.} We show that \textsc{Discrete-G-Min-Disc-Code} in 1D using interval objects of arbitrary length, is NP-complete by using a polynomial time reduction from a restricted version of the 3-SAT problem. Here, the challenge is to overcome the linear nature of the problem and to transmit the information across the entire construction without affecting intermediate regions. This result is in contrast with \textsc{Continuous-G-Min-Disc-Code} in 1D, which is polynomial-time solvable~\cite{GledelP19}. This is also in contrast with most geometric covering problems, which are often polynomial-time solvable in 1D~\cite{covering}.

We then design a simple polynomial-time $2$-factor approximation algorithm for \textsc{Discrete-G-Min-Disc-Code} in 1D, that is much simpler than that published in the preliminary version of this paper~\cite{isaac}.

We also design a PTAS for both \textsc{Discrete-G-Min-Disc-Code} and \textsc{Continuous-G-Min-Disc-Code} in 1D, when all the objects are required to have the same (unit) length. In this context, it needs to be mentioned once again that \textsc{Continuous-G-Min-Disc-Code} for arbitrary intervals is polynomially solvable~\cite{GledelP19}.  

We also study both problems in 2D for axis-parallel unit square objects, which form a natural extension of 1D intervals to the 2D setting. The continuous version is known to be NP-complete for unit disks~\cite{GledelP19}, and we show that the reduction can be adapted to our setting, for both the continuous and discrete cases.

We then design polynomial-time constant-factor approximation algorithms for both problems in that setting. The approximation factors are $16$ and $64$ for the continuous and discrete problem respectively.\footnote{Note that, the $(4+\epsilon)$ factor approximation algorithms presented in the conference version of this paper~\cite{isaac} were wrong and the algorithms have been corrected here.}
To obtain these algorithms, we re-formulate our problems into a problem of \emph{stabbing} line segments in 2D, which can be reduced to a geometric \textsc{Hitting Set} problem. 

Our results on Discriminating Code problems are summarized in Table~\ref{table}. 

\renewcommand{\arraystretch}{2}
\begin{table}
	\centering
	\scalebox{0.8}{
		\begin{tabular}{p{3cm}|c|c|c|c}
			\multirow{2}{2.2cm}{\textsc{Object Type}} & 
			\multicolumn{2}{c}{\textsc{Continuous-G-Min-Disc-Code}} & 
			\multicolumn{2}{c}{\textsc{Discrete-G-Min-Disc-Code}}\\
			\cline{2-5}
			& \textsc{Hardness} & \textsc{Algorithm} & \textsc{Hardness} & 
			\textsc{Algorithm}\\
			\hline
			1D intervals & - & Polynomial-time solvable (\cite{GledelP19}) 
			& NP-hard (Thm.~\ref{thm:1D-NPC}) & $2$-approximable (Thm.~\ref{thth2}) \\ \hline

			1D bounded intervals & - & Polynomial-time solvable (\cite{GledelP19}) 
			& Open & $2$-approximable (Thm.~\ref{thth2}) \\ \hline

			1D unit intervals & Open & PTAS (Cor.~\ref{cor:PTAS}) & Open & PTAS 
			(Thm.~\ref{thm:PTAS}) \\ \hline

			2D axis-parallel unit squares & NP-hard (Thm.~\ref{unitsquares-NPc})& 
			$(16\cdot OPT + 1)$-approximable (Thm.~\ref{8approx}) & NP-hard (Thm.~\ref{unitsquares-NPc})&  $(64\cdot OPT + 1)$-approximable (Thm.~\ref{Dapprox}) \\
	\end{tabular}}
	
	\caption{Summary of our results on \textsc{G-Min-Disc-Code} problems.}
	\label{table}
\end{table}

Finally, we consider the related \textsc{Min-Id-Code} problem restricted to unit square graphs (geometric intersection graphs for 2D axis-parallel unit squares).  
We show that \textsc{Min-Id-Code} for unit square graphs can be solved in the same manner as  the \textsc{Discrete-G-Min-Disc-Code} problem for axis-parallel unit square objects, and our approximation results for \textsc{Discrete-G-Min-Disc-Code} still hold for \textsc{Min-Id-Code} on this class of graphs.

\noindent{\bf Structure of the paper.} We start with our results about \textsc{G-Min-Disc-Code} in 1D in Section~\ref{sec:1D}. We then present our results on \textsc{G-Min-Disc-Code} for axis-parallel unit squares in 2D in Section~\ref{sec:2D}. Our results about \textsc{Min-ID-Code} for unit square intersection graph is presented in Section~\ref{sec:IDcode}. Finally, we conclude the paper in Section~\ref{sec:conclu}.

\section{The \textsc{G-Min-Disc-Code} problem in 1D}\label{sec:1D}

It has been shown that \textsc{Continuous-G-Min-Disc-Code} is polynomial-time solvable in 1D~\cite{GledelP19}. Thus, in this section we focus on \textsc{Discrete-G-Min-Disc-Code}.

An instance $(P,S)$ of the \textsc{Discrete-G-Min-Disc-Code} problem is a set $P=\{p_1, \ldots, p_n\}$ of points and a set $S$ of $m$ intervals of arbitrary lengths placed on a real line $\IR$. Assuming that the points are sorted with respect to their $x$-coordinate values, we define $n+1$ {\it gaps} ${\cal G} = \{g_1, \ldots, g_{n+1}\}$, where $g_1=(-\infty,p_1)$, $g_i=(p_{i-1},p_i)$ for $2\leq i\leq n$, and $g_{n+1}=(p_n,\infty)$.

Observe that (i) if both endpoints of an interval $s \in S$ lie in the same gap of $\cal G$, then it can not  discriminate any pair of points; thus $s$ is \emph{useless}, and (ii) if more than one interval in $S$ have both their endpoints in the same two gaps, say $g_a = (p_a,p_{a+1}), g_b=(p_b,p_{b+1}) \in \cal G$, then both of them discriminate exactly the same point-pairs. Thus, they are \emph{redundant} and we need to keep only one interval among them. In a linear scan, we can first eliminate the useless and redundant intervals. From now onwards, $m$ will denote the number of intervals, none of which are useless or redundant. Hence, $m$ may be $O(n^2)$ in the worst case.

\subsection{NP-completeness}\label{sec:NP-c-1D}

The \textsc{Discrete-G-Min-Disc-Code} problem is in NP, since given a subset $S' \subseteq S$, one can test in polynomial time whether the problem instance $(P,S')$ is twin-free (i.e., whether the id of every point in $P$ induced by $S'$ is unique) by Proposition~\ref{twin-free}. 

We prove the NP-hardness of \textsc{Discrete-G-Min-Disc-Code} using a polynomial-time reduction from the 3-SAT-$2l$ problem~(defined below).

\Pb{3-SAT-$2l$}
{A collection of $m$ clauses $C = \{c_1, c_2, \ldots, c_m\}$ where each clause contains at most three literals, over a set of $n$ Boolean variables $X = \{x_1, x_2, \ldots, x_n\}$, and each literal appears at most twice.}
{A truth assignment of $X$ such that each clause is satisfied (if it exists).}

\begin{theorem}[\cite{SAT}]
  3-SAT-$2l$ is NP-complete.
\end{theorem}
\begin{proof}
\green{It is stated in~\cite[Theorem~2.1]{SAT} that ``Boolean satisfiability is NP-complete when restricted to instances with 2 or 3 variables per clause and at most 3 occurrences per variable''. The following simple argument can be used to derive the statement. Consider a formula with 2 or 3 variables per clause and at most 3 occurrences per variable. If the same literal appears exactly three times in the formula, it means its negation never appears, so we might as well set its variable so that this literal is true, and remove all clauses containing that variable. We then get an equivalent formula. By doing this repeatedly, we obtain an equivalent instance where each literal appears at most twice. Thus, Theorem~2.1 from~\cite{SAT} implies that 3-SAT-$2l$ is NP-complete.\hfill\qed}
\end{proof}



Given an instance $(X,C)$ of 3-SAT-$2l$, we construct in polynomial time an instance $(P,S)=\Gamma(X, C)$ of the \textsc{Discrete-G-Min-Disc-Code} problem on the real line $\mathbb{R}$. The main challenge of this reduction is to be able to connect variable and clause gadgets, despite the linear nature of our 1D setting. The basic idea is that we will construct an instance where some specific set of \emph{critical} point-pairs will need to be discriminated (all other pairs being discriminated by some partial solution forced by our gadgets). Let us start by 
describing our basic gadgets.

\begin{definition}\label{d1}
A \emph{covering gadget} $\Pi$ consists of three intervals $I$, $J$, $K$ and four points $p_1$, $p_2$, $p_3$ and $p_4$ satisfying $p_1\in I$, $p_2 \in I\cap J$, $p_3 \in I \cap J \cap K$ and $p_4 \in J \cap K$ as in Figure~\ref{fig:coverG}. Every other interval of the construction will either contain all four points, or none. There may exist a set of points in $K\setminus\{I\cup J\}$, depending on the need of the reduction.
\end{definition}

\begin{observation}\label{luv}
The points $p_1$,$p_2$,$p_3$,$p_4$ can only be discriminated by choosing all three intervals $I$, $J$, $K$ in the solution. 
\end{observation}

\begin{proof}
Follows from the fact that none of the intervals in $\Gamma(X, C)$, that is not a member of the covering gadget $\Pi$ can discriminate the four points in $\Pi$. Moreover, if we do not choose $I$, then $p_3,p_4$ are not discriminated. If we do not choose $J$, $p_1,p_2$ are not discrimnated. If we do not choose $K$, $p_2,p_3$ are not discriminated (see Figure~\ref{fig:coverG}).\hfill \qed
\end{proof}

The idea of the covering gadget is to forcefully cover the points placed in $K\setminus\{I\cup J\}$, so that they are covered by $K$ (which needs to be in any solution), and hence discriminated from all other points of the construction.
	
	\begin{figure}[ht!]
		\centering
		\includegraphics[scale=0.9]{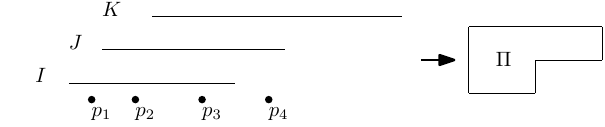}
		\caption{A covering gadget $\Pi$, and its schematic representation.} 
		\label{fig:coverG}
\end{figure}

Let us now define the gadgets modeling the clauses and variables of the 3-SAT-$2l$ instance.

\begin{definition}\label{cl-gadget}
Let $c_i$ be a clause of $C$. The \emph{clause gadget} for $c_i$, denoted $G_c(c_i)$, is defined by a covering gadget $\Pi(c_i)$ along with two points $p_{c_i}, p'_{c_i}$ placed in $K\setminus\{I\cup J\}$ (see Figure~\ref{fig:clauseG}).
\end{definition}

	\begin{figure}[ht!]
		\centering
		\includegraphics[scale=0.75]{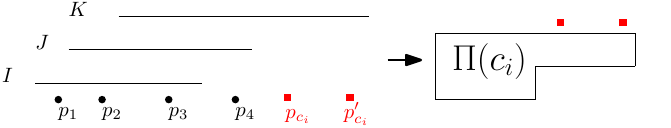}
		\caption{A clause gadget $\Pi(c_i)$, and its schematic representation.} 
		\label{fig:clauseG}
\end{figure}

The idea behind the clause gadget is that some interval that ends between points $p_{c_i}, p'_{c_i}$ will have to be taken in the solution, so that this pair gets discriminated.

\begin{definition}\label{var-gadget}
	Let $x_j$ be a variable of $X$. The \emph{variable gadget} for $x_j$, denoted $G_v(x_j)$, is defined by a covering gadget $\Pi(x_j)$, and five points $p^1_{x_j},\ldots, p^5_{x_j}$ placed consecutively in $K\setminus\{I\cup J\}$. We place six intervals $I_{x_j}^0$, $I_{x_j}^1$, $I_{x_j}^2$, $I_{\overline{x_j}}^0$, $I_{\overline{x_j}}^1$, $I_{\overline{x_j}}^2$, as in Figure~\ref{fig:varG}.
	      
	\begin{itemize}
		\item Interval $I_{x_j}^0$ starts between $p^1_{x_j}$ and $p^2_{x_j}$, and ends between $p^3_{x_j}$ and $p^4_{x_j}$.
		\item Interval $I_{\overline{x_j}}^0$ starts between $p^2_{x_j}$ and $p^3_{x_j}$, and ends between $p^4_{x_j}$ and $p^5_{x_j}$.
		\item Interval $I_{x_j}^1$ starts between $p^2_{x_j}$ and $p^3_{x_j}$, and ends after $p^5_{x_j}$.
		\item Interval $I_{x_j}^2$ starts between $p^4_{x_j}$ and $p^5_{x_j}$, and ends after $p^5_{x_j}$.
		\item Interval $I_{\overline{x_j}}^1$ starts between $p^1_{x_j}$ and $p^2_{x_j}$, and ends after $p^5_{x_j}$.
		\item Interval $I_{\overline{x_j}}^2$ starts between $p^3_{x_j}$ and $p^4_{x_j}$, and ends after $p^5_{x_j}$.
	\end{itemize}

	(The end-point of the four intervals, namely $I_{x_j}^1$, $I_{\overline{x_j}}^1$, $I_{x_j}^2$, $I_{\overline{x_j}}^2$, will be determined at the time of construction of the whole instance.)
\end{definition}	

	\begin{figure}[ht!]
		\centering
		\includegraphics[scale=0.6]{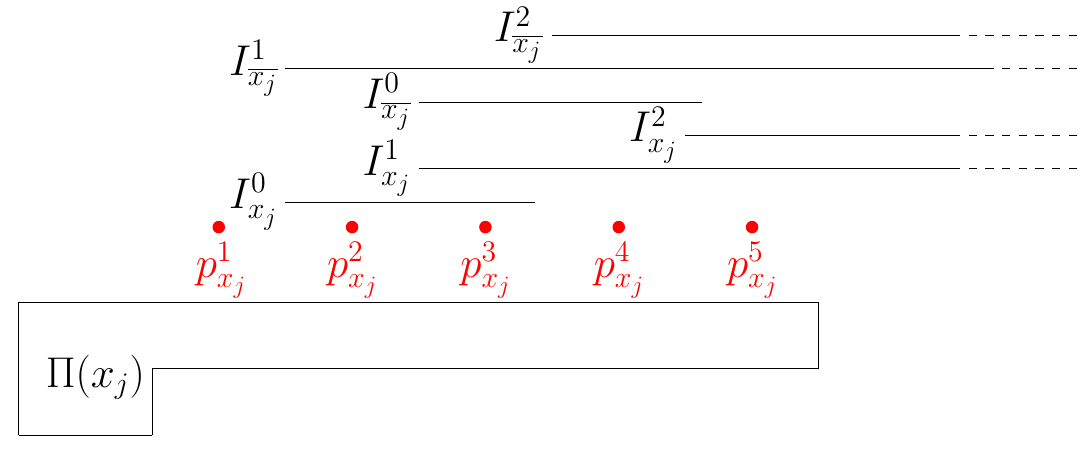}
		\caption{Variable gadget for variable $x_j$.} 
		\label{fig:varG}
	\end{figure}
\remove{ CONFVERSION
\begin{figure}[ht!]
	\centering
	\begin{minipage}{0.25\textwidth}
		\centering
		\includegraphics[scale=0.49]{combineG.pdf}
	\end{minipage}\hspace{0.08\textwidth}
	\begin{minipage}{0.6\textwidth}
		\centering
		\includegraphics[scale=0.37]{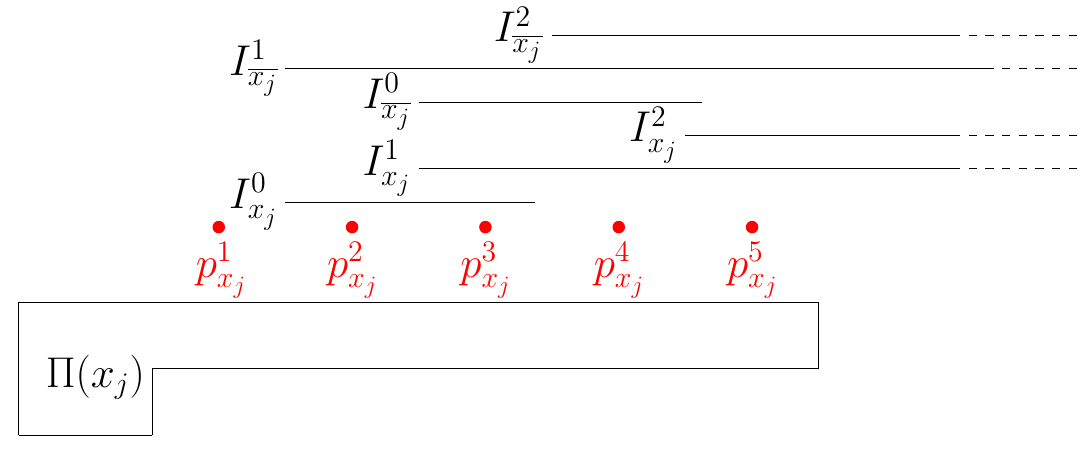}
	\end{minipage}
	\caption{(a) A covering gadget $\Pi$, and its schematic representation, (b) A 
	clause gadget $G_c(c_i)$, (c) A variable gadget $G_v(x_j)$}
	\label{fig:gadget}	
\end{figure}}

In a variable gadget $G_v({x_j})$, the intervals $I_{x_j}^1$ and $I_{x_j}^2$ represent the two possible occurrences of literal ${x_j}$, while $I_{\overline{x_j}}^1$ and $I_{\overline{x_j}}^2$ represent the two possible occurrences of $\overline{x_j}$. The right end points of each of these four intervals will be in the clause gadget of the clause where the occurrence of that literal takes place. An example for the construction of $\Gamma(X, C)$ is shown in Figure~\ref{fig:setUp}. We assume that every literal appears in at least one clause\footnote{If a literal does not appear in any clause, then we can assign a truth value to its variable so that all its occurrences are true, and further ignore this variable.}.

\begin{figure}[htpb!]
	\centering
	\includegraphics[scale=0.8,angle=90]{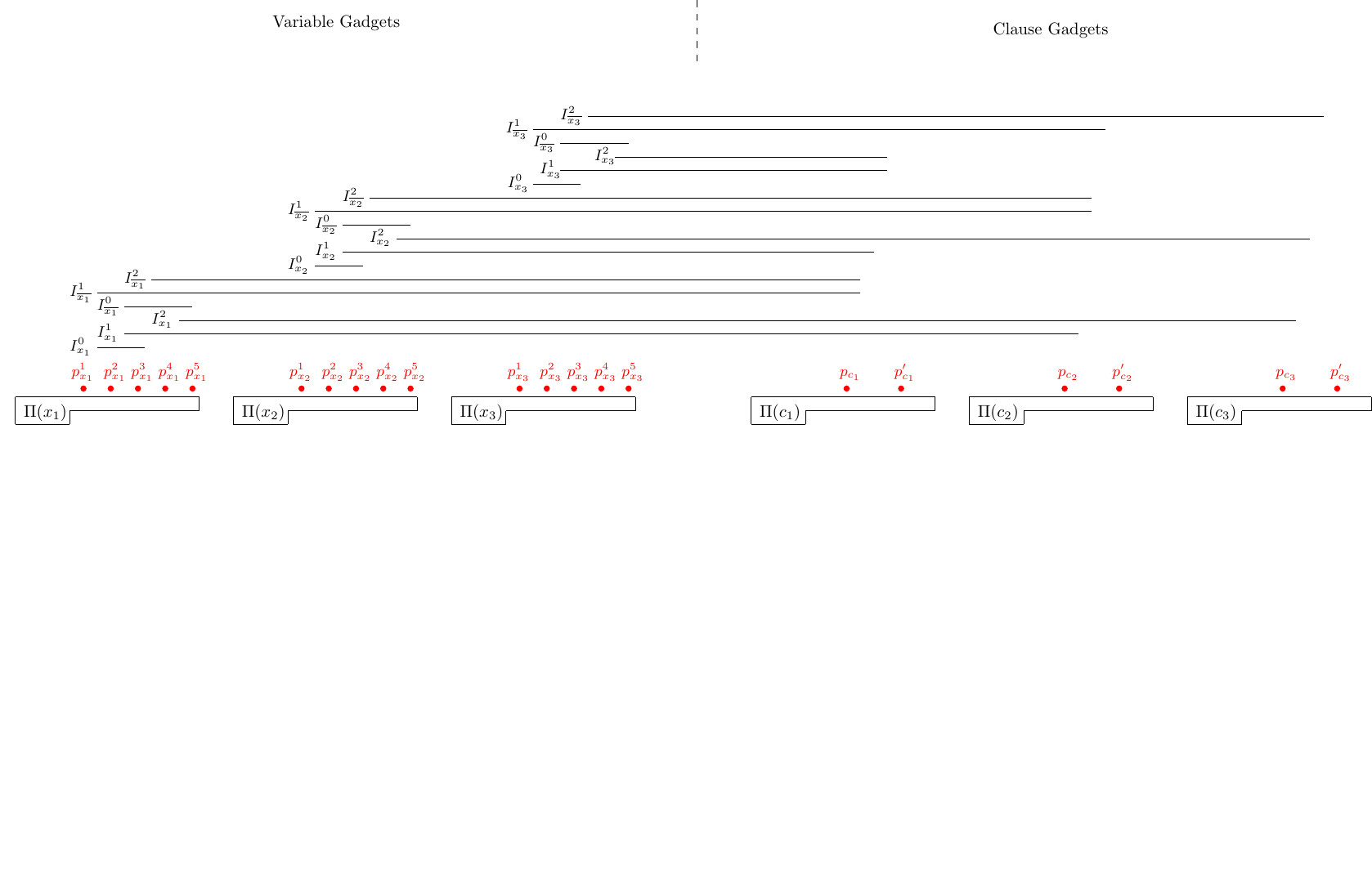}\\
	\caption{The instance $\Gamma(X, C)$ for the formula $(X,C)=(\overline{x_1} \lor x_2 \lor x_3) \land (x_1 \lor \overline{x_2} \lor \overline{x_3}) \land (x_1 \lor x_2 \lor x_3)$.} 
	\label{fig:setUp}
\end{figure}

\begin{itemize}
	\item For each variable $x_i \in X$, $\Gamma(X, C)$ contains a variable gadget $G_v(x_i)$.
	\item The gadgets $G_v(x_1), G_v(x_2), \dots , G_v(x_n)$ are positioned consecutively, in this order, without overlap.
	\item For each clause $c_j \in C$, $\Gamma(X, C)$ contains a clause gadget $G_c(c_j)$.
	\item The gadgets $G_c(c_1), G_c(c_2), \dots, G_c(c_m)$ are positioned consecutively, in this order, to the right of the placement of the variable gadgets, without overlap.
	\item For every variable $x_i$, assume $x_i$ appears in clauses $c_{i_1}$ and $c_{i_2}$, and $\overline{x_i}$ appears in $c_{i_3}$ and $c_{i_4}$ (possibly $i_1=i_2$ or $i_3=i_4$). Then, we extend the intervals $I^1_{x_i}$, $I^1_{\overline{x_i}}$, $I^2_{x_i}$, $I^2_{\overline{x_i}}$ so that $I^1_{x_i}$ ends between $p_{c_{i_1}}$ and $p'_{c_{i_1}}$, $I^2_{x_i}$ ends between $p_{c_{i_2}}$ and $p'_{c_{i_2}}$, $I^1_{\overline{x_i}}$ ends between $p_{c_{i_3}}$ and $p'_{c_{i_3}}$, and $I^2_{\overline{x_i}}$ ends between $p_{c_{i_4}}$ and $p'_{c_{i_4}}$.
\end{itemize}

Let ${\mathcal C}^{\Pi}$ be the union of the discriminating codes (i.e., all intervals of type $I,J,K$, as mentioned in Observation~\ref{luv}) of all covering gadgets of type $\Pi(x_i)$ of $\Gamma(X, C)$.

Consider any covering gadget $\Pi_1$. By Observation~\ref{luv}, the points $p_1,p_2,p_3,p_4$ are discriminated among each other by the set ${\mathcal C}^{\Pi}$, and they are discriminated from all other points of $\Gamma(X, C)$, since they are the only ones to be covered by one of the intervals $I,J$ from $\Pi_1$. Moreover, all the points covered by the interval $K$ from $\Pi_1$ are discriminated from all the points not covered by $K$. Thus, overall, all point-pairs of $\Gamma(X, C)$ are discriminated by ${\mathcal C}^{\Pi}$, except the following \emph{critical} point-pairs:
\begin{itemize}
	\item the point-pairs among the five points $p^1_{x_i},\ldots, p^5_{x_i}$ of each variable gadget $G_v(x_i)$, and
	\item the point-pair $\{p_{c_j},p'_{c_j}\}$ of each clause gadget $G_c(c_j)$.
\end{itemize}
In the proof of the following main result of this section, we will demonstrate, in particular, that if there exists a truth assignment of the variables in $X$ such that all the clauses in $C$ are satisfied, then the critical point-pairs are also discriminated.

\begin{theorem}\label{thm:1D-NPC}
	\textsc{Discrete-G-Min-Disc-Code} in 1D is NP-complete.
\end{theorem}
\begin{proof}
	We prove that $(X, C)$ is satisfiable if and only if $\Gamma(X, C)$ has a discriminating code of size $6n + 3m$. 	In both parts of the proof, we will consider the set $\mathcal{C}^{\Pi}$ defined above. Each variable gadget and clause gadget contains one covering gadget. Thus, $|\mathcal{C}^{\Pi}| = 3(n + m)$.
	
	Consider first some satisfying truth assignment of $X$. We build a solution set $\mathcal C$ as follows. First, we put all intervals of $\mathcal{C}^{\Pi}$ in $\mathcal{C}$. Then, for each variable $x_i$, if $x_i$ is true, we add intervals $I_{x_i}^0$, $I_{x_i}^1$ and $I_{x_i}^2$ to $\mathcal C$. Otherwise, we add intervals $I_{\overline{x_i}}^0$, $I_{\overline{x_i}}^1$ and $I_{\overline{x_i}}^2$ to $\mathcal C$. Notice that $|\mathcal{C}| = 6n + 3m$. As observed before, it suffices to show that $\mathcal C$ discriminates the point-pair $\{p_{c_j},p'_{c_j}\}$ of each clause gadget $G_c(c_j)$, and the points $p^1_{x_i},\ldots, p^5_{x_i}$ of each variable gadget $G_v(x_i)$. (All other pairs are discriminated by ${\mathcal C}^{\Pi}$.)
	
	Since the assignment is satisfying, each clause $c_j$ contains a true literal $l_i\in\{x_i,\overline{x_i}\}$. Then, one interval of $G_v(x_i)$ is in $\mathcal C$ and discriminates $p_{c_j}$ and $p'_{c_j}$. Furthermore, consider a variable $x_i$. Point $p^1_{x_i}$ is discriminated from $p^2_{x_i},\dots,p^5_{x_i}$ as it is the only one not covered by any of $I_{x_i}^0$, $I_{x_i}^1$, $I_{x_i}^2$, $I_{\overline{x_i}}^0$, $I_{\overline{x_i}}^1$, and $I_{\overline{x_i}}^2$. If $x_i$ is true, $p^2_{x_i}$ is covered by $I_{x_i}^0$; $p^3_{x_i}$ is covered by $I_{x_i}^0$ and $I_{x_i}^1$; $p^4_{x_i}$ is covered by $I_{x_i}^1$; $p^5_{x_i}$ is covered by $I_{x_i}^1$ and $I_{x_i}^2$. If $x_i$ is false, $p^2_{x_i}$ is covered by $I_{\overline{x_i}}^1$; $p^3_{x_i}$ is covered by $I_{\overline{x_i}}^0$ and $I_{\overline{x_i}}^1$; $p^4_{x_i}$ is covered by $I_{\overline{x_i}}^0$, $I_{\overline{x_i}}^1$ and $I_{\overline{x_i}}^2$; $p^5_{x_i}$ is covered by $I_{\overline{x_i}}^1$ and $I_{\overline{x_i}}^2$. Thus, in both cases, the five points are discriminated, and $\mathcal C$ is discriminating, as claimed.

	For the converse, assume that $\mathcal C$ is a discriminating code of $\Gamma(X, C)$ of size $6n+3m$. By Observation~\ref{luv}, ${\mathcal C}^{\Pi} \subseteq \mathcal C$. Thus there are $3n$ intervals of $\mathcal C$ that are not in ${\mathcal C}^{\Pi}$.
	
	First, we show that $\mathcal C\setminus {\mathcal C}^{\Pi}$ contains exactly three intervals of each variable gadget $G_v(x_i)$. Indeed, it cannot contain less than three, otherwise we show that the points $p^1_{x_i},\ldots, p^5_{x_i}$ cannot be discriminated. To see this, note that each consecutive pair $\{p^s_{x_i},p^{s+1}_{x_i}\}$ ($1\leq s\leq 4$) must be discriminated, thus $\mathcal C$ must contain one interval with an endpoint between these two points. There are four such consecutive pairs in $G_v(x_i)$, thus if $\mathcal C\setminus {\mathcal C}^{\Pi}$ contains at most two intervals of $G_v(x_i)$, it must contain $I_{x_i}^0$ and $I_{\overline{x_i}}^0$. But now, the points $p^1_{x_i}$ and $p^5_{x_i}$ are not discriminated, a contradiction.
	
	Let us now show how to construct a truth assignment of $(X,C)$. Notice that at least one of $I_{x_i}^0$ and $I_{\overline{x_i}}^0$ must belong to $\mathcal C$, otherwise some points of $G_v(x_i)$ cannot be discriminated. If $I_{x_i}^0\in \mathcal C$ but $I_{\overline{x_i}}^0\notin C$, then necessarily $I_{x_i}^1\in \mathcal C$ to discriminate $p^2_{x_i}$ and $p^3_{x_i}$, and $I_{x_i}^2\in \mathcal C$ to discriminate $p^4_{x_i}$ and $p^5_{x_i}$. In this case, we set $x_i$ to true. Similarly, if $I_{\overline{x_i}}^0\in \mathcal C$ but $I_{x_i}^0\notin \mathcal C$, then necessarily $I_{\overline{x_i}}^1\in \mathcal C$ to discriminate $p^1_{x_i}$ and $p^2_{x_i}$, and $I_{\overline{x_i}}^2\in \mathcal C$ to discriminate $p^3_{x_i}$ and $p^4_{x_i}$. In this case, we set $x_i$ to false. Finally, if both $I_{x_i}^0$ and $I_{\overline{x_i}}^0$ belong to $\mathcal C$, the third interval of $\mathcal C\setminus {\mathcal C}^{\Pi}$ in $G_v(x_i)$ may be any of the four intervals covering $p^5_{x_i}$. If this third interval is $I_{x_i}^1$ or $I_{x_i}^2$, we set $x_i$ to true; otherwise, we set it to false.
	
	Observe that when we set $x_i$ to true, none of $I_{\overline{x_i}}^1$ and $I_{\overline{x_i}}^2$ belongs to $\mathcal C$; likewise, when we set $x_i$ to false, none of $I_{x_i}^1$ and $I_{x_i}^2$ belongs to $\mathcal C$. Thus, our truth assignment is coherent. As for every clause $c_j$, the point-pair $\{p_{c_j},p'_{c_j}\}$ is discriminated by $\mathcal C$, one interval correspoding to a true literal discriminates it. The obtained assignment is satisfying, completing the proof.\hfill \qed
\end{proof}

\subsection{A $2$-approximation algorithm}	\label{sec:2-approx}
We next design a 2-approximation algorithm for \textsc{Discrete-G-Min-Disc-Code} in 1D by carefully choosing at most $n$ intervals to discriminate the $n$ points of $P$. We remark that this algorithm is substantially simpler than the one from the preliminary version of this paper~\cite{isaac}.

First, we will need the following proposition, already observed in~\cite{GledelP19}.

\begin{proposition}[\cite{GledelP19}]\label{prop:n/2}
Any solution of \textsc{Discrete-G-Min-Disc-Code} and \textsc{Continuous-G-Min-Disc-Code} in 1D for inputs of $n$ points has size at least $\frac{n+1}{2}$.
\end{proposition}
\begin{proof}
  In order to discriminate the consecutive points in $P$, for any feasible solution $SOL$ for $P$, every gap between two consecutive points will contain an end-point of at least one interval in $SOL$. There exist $n-1$ gaps for the $n$ points in $P$. But we must also have intervals covering the first point and the last point, which amounts to $n+1$ positions for the end-points of intervals of $SOL$. Thus, any solution has size at least $(n+1)/2$.\qed
\end{proof}

\begin{theorem} \label{thth2}
	There exists a $2$-factor approximation algorithm solving \textsc{Discrete-G-Min-Disc-Code} in 1D, that runs in $O(m\log m)$ time and $O(m)$ space.
\end{theorem}
\begin{proof}
We assume that the points of $P=\{p_1, \ldots, p_n\}$ of the instance $(P,S)$ are sorted in increasing order with respect to their $x$-coordinate values. At Step~$i$, our partial solution $S_i\subseteq S$ will have the property that it covers and discriminates all the points in $\{p_1,\ldots,p_i\}$ (using at most $i$ intervals). Thus, at step~$n$, $S_n$ is a valid solution for $(P,S)$ of size at most $n$.

For the first step, we let $S_1$ consist of any interval that contains $p_1$. For the next steps, we assume that $i$ iterations have already been executed, and thus we have computed the set $S_{i}$ that discriminates $P_i=\{p_1,\ldots,p_{i}\}$. We now consider the set $P_{i+1}=\{p_1,\ldots,p_{i+1}\}$. We distingish three cases as follows.

\begin{description}
  	\item[Case 1:] The id of $p_{i+1}$ using the intervals in $S_i$ is 	non-null and is different from the id of every point $p \in P_i$. Here, simply let $S_{i+1}=S_i$, that is, $S_i$ is already a feasible solution for $P_{i+1}$.
	\item[Case 2:] The id of $p_{i+1}$ using the intervals in $S_i$ is null, that is, no interval of $S_i$ covers $p_{i+1}$. Here, we choose any arbitrary interval $s \in S$ that covers $p_{i+1}$. Thus, we have $S_{i+1}=S_i \cup \{s\}$. As the members in $P_i$ are already discriminated up to $i$-th iteration, they remain discriminated by $S_{i+1}$ (even if the new interval covers a subset of $P_i$). 
	\item[Case 3:] The id of $p_{i+1}$ using the intervals in $S_i$ is non-null, but is the same as the id of some $p_j \in P_i$ $(j < i)$. Note that, in such a case the id of $p_{i+1}$ can match with at most one element of $P_i$ as the id's of $P_i$ are all distinct with respect to $S_i$. Here, we choose an interval $s$ of $S$ that can discriminate $p_j$ and $p_{i+1}$. Such an interval always exists as we have already checked that $(P,S)$ is twin-free. Thus, $S_{i+1}=S_i \cup \{s\}$ is our valid partial solution.  
\end{description}

As we have inserted at most one interval at each iteration, $|S_n|\leq n$. By Proposition~\ref{prop:n/2}, $|S_n|\leq 2 OPT-1$ and the $2$-approximation factor follows.


We now analyze the time and space complexity. We will use two tree data structures for processing the points in $P$ and the intervals in $S$ efficiently. These are a \emph{height-balanced binary tree} ${\cal T}_H$, and a \emph{priority search tree} ${\cal T}_P$. 

\begin{description}
 \item[${\cal T}_H$:] It is a binary tree in which the depth of the two subtrees of every node does not differ by more than 1~\cite{knuth97}. A height-balanced binary tree with $n$ nodes has height $\Theta(\log n)$. Each operation (lookup, insertion or deletion) takes time $\Theta(\log n)$ in the worst case.
 \item[${\cal T}_P$:] A priority search tree \cite{PrioritySearchTree,Berg} is a hybrid of a priority queue and a binary search tree. It stores a set of 2-dimensional points (a pair of real numbers) for the efficient answering of 1.5-dimensional queries in a one-side open query box of the form $(-\infty,a]\times [b,c]$. In other words, it can report/count the points whose $x$-coordinate is smaller than $a$, and $y$-coordinate lies in the range $[b,c]$. The preprocessing time and space complexities of this data structure are $O(n\log n)$ and $O(n)$ respectively; the time complexity for reporting/counting a 1.5-dimensional query is $O(s+\log n)$/$O(\log n)$, where $s$ is the number of points returned by the search.
\end{description}

Before the start of the algorithm, we compute a  ${\cal T}_P$ data structure with a set of pairs of reals $(\ell(s),r(s))$ corresponding to the segments $S$, where $\ell(s)$ and $r(s)$ denote the coordinates of the left and right end-point of the interval $s\in S$ (on the $x$-axis). The preprocessing time and space required for ${\cal T}_P$ are $O(m\log m)$ and $O(m)$ ($m=|S|$) respectively. Identifying the intervals in $S$ that contains a point $p\in P$ and does not contain a point $q \in P$ is equivalant to a 1.5-dimensional range query with the query box $(-\infty,x(p)] \times (x(p),x(q))$, where $x(p)$ denotes the x-coordinate of the point $p$.

The height-balanced binary tree ${\cal T}_H$ stores the {\it groups} generated after processing the intervals in $S_i$, and is updated at the end of each iteration $i$. A {\em group} is a maximal set of pairwise intersecting intervals of $S_i$. These groups are totally ordered in the sense that a pair of consecutive groups share only their common end-point. Each group contains at most one point of $P_i$. Since $|S_i|\leq i$, and the number of groups created with $|S_i|$ is $2\times |S_i|+1$, the size of the data structure ${\cal T}_H$ for storing the groups during the entire execution is $O(n)$. While processing $p_{i+1}$, this tree structure is used to identify an appropriate interval to cover the point $p_{i+1}$ in $O(\log n)$ time. As a result, we also know which case to follow for discriminating $p_{i+1}$. The cost of  maintenance of ${\cal T}_H$ after each iteration is also $O(\log n)$

If Case~2 happens, we need to identify an interval $s \in S$ that covers $p_{i+1}$, i.e., a segment with $\ell(s) < x(p_{i+1}) < r(s)$, or in other words, a pair of reals $(\ell(s),r(s))$ that lies in the range box $(-\infty,x(p_{i+1})]\times [x(p_{i+1}),\infty)$. 

If Case~3 happens, then we need to choose a segment $s \in S$ satisfying 

either $x(p_j) < \ell(s) < x(p_{i+1}) < r(s)$ i.e., $(\ell(s),r(s))\in (x(p_j),x(p_{i+1})] \times [x(p_{i+1}), \infty)$  

or $\ell(s) < x(p_j) < r(s) < x(p_{i+1})$ i.e., $(\ell(s),r(s))\in(-\infty,x(p_j)]\times [x(p_j),x(p_{i+1}))$. 

As mentioned earlier, in either of the cases such an element can be found in the data structure ${\cal T}_P$ in $O(\log m)$ time. Thus, the only task that remains is to modify the data structure ${\cal T}_H$ after inserting $s=[a,b]\in {\cal T}_H$. Let $a$ and $b$ lie in the groups $g_\alpha=[\theta_1,\theta_2]$ and $g_\beta=[\psi_1,\psi_2]$, respectively. Now, $g_\alpha$ and $g_\beta$ is to be deleted from ${\cal T}_H$ and four intervals $[\theta_1,a]$, $[a,\theta_2]$, $[\psi_1,b]$ and $[b,\psi_2]$ need to be inserted in ${\cal T}_H$. If $=[a,b]$ lies entirely in the same group $[\theta_1,\theta_2]$, then $[\theta_1,\theta_2]$ is split into three groups $[\theta_1,a]$, $[a,b]$ and $[b,\theta_2]$. Surely, we need to attach $p_j$ and $p_{i+1}$ to the appropriate interval groups to which they belong. This needs another $O(\log n)$ time. Thus, processing $p_{i+1}$ requires $O(\log m+\log n)$ time in the worst case. 

The construction of ${\cal T}_P$ needs $O(m\log m)$ time and $O(m)$ space~\cite{Lee04,PrioritySearchTree}. Processing $n$ points requires $O(n(\log n+\log m))$ time as explained in the previously. As $n=O(m)$, the result follows. \hfill\qed
\end{proof}

\subsection{A PTAS for the unit interval case}\label{sec:PTAS}

	We now design a PTAS for the 1D case where all intervals in $S$ have the same length.

	The following observation (which was also made in the related setting of identifying codes of unit interval graphs~\cite[Proposition 5.12]{Foucaud}) plays an important role in designing our PTAS.

	\begin{observation} \label{imp}
		In an instance $(P,S)$ of \textsc{Discrete-G-Min-Disc-Code} in 1D, if the objects in $S$ are intervals of the same length, then discriminating all the pairs of \emph{consecutive} points in $P$ is equivalant to discriminating \emph{all} the pairs of points in $P$.
	\end{observation}

	\begin{proof}
		Assume that we have a set $S'\subseteq S$ that covers all points and discriminates all consecutive point-pairs, but two non-consecutive points $p_i$ and $p_j$ ($i<j$) are not discriminated. Since $p_i$ and $p_j$ are covered by the same set of intervals of $S'$ and the intervals are of same (unit) length, they must be at a distance at most~$1$ apart. Now, since they are not consecutive, $p_{i+1}$ lies between $p_i$ and $p_j$. Since $S'$ discriminates $p_i$ and $p_{i+1}$, there is an interval $I\in S'$ with an endpoint in the gap $g_i=[p_i,p_{i+1}]$. If it is the right endpoint, $I$ covers $p_i$ but not $p_j$, a contradiction. Thus, it must be the left endpoint. But since the distance between $p_i$ and $p_j$ is at most~$1$, $I$ contains $p_j$ (but not $p_i$), again a contradiction. \hfill \qed
	\end{proof}

	For a given $\epsilon > 0$, we choose $\lceil\frac{n\epsilon}{4}\rceil$ points, namely $q_1, q_2, \ldots, q_{\lceil\frac{n\epsilon}{4}\rceil} \in P$, called the {\it reference points}, as follows: $q_1$ is the $\lceil\frac{2}{\epsilon}\rceil$-th point of $P$ from the left, and for each $i=1,2, \ldots, \lfloor\frac{n\epsilon}{4}\rfloor$, the number of points in $P$ between every consecutive pair $(q_i, q_{i+1})$ is $\lceil\frac{4}{\epsilon}\rceil$ (both inclusive). The number of points to the right of $q_{\lceil\frac{n\epsilon}{4}\rceil}$ may be less than $\lceil\frac{2}{\epsilon}\rceil$. For each {\it reference point} $q_i$, we choose two intervals $I_i^1, I_i^2 \in S$ such that both $I_i^1, I_i^2$ contain (span) $q_i$,  and the left (resp. right) endpoint of $I_i^1$ (resp. $I_i^2$) have the minimum $x$-coordinate (resp. maximum $x$-coordinate) among all intervals in $S$ that span $q_i$. Observe that all the points in $P$ that lie in the range $G_i=[\ell(I_i^1), r(I_i^2)]$ are \emph{covered}, where $\ell(I_i^1)$ and $r(I_i^2)$ are the $x$-coordinates of the left endpoint of $I_i^1$ and the right endpoint of $I_i^2$, respectively. These ranges will be referred to as \emph{group-ranges}. Since the endpoints of the intervals are distinct, the span of a \emph{group-range} is strictly greater than~1.

	We now define a \emph{block} as follows. Observe that the ranges $G_i$ and $G_{i+1}$ may or may not overlap. If several consecutive ranges $G_i, G_{i+1}, \ldots, G_k$ are pairwise overlapping, then the horizontal range $[\ell(I_i^1), r(I_k^2)]$ forms a block. The region between a pair of consecutive blocks will be referred to as a \emph{free region}. We use $B_1, B_2, \ldots, B_l$ to name the blocks in order, and $F_0, F_1, \ldots,F_l$ to name the free regions (from left to right). The points in each block are covered. Here, the remaining tasks are (i) for each block, choose intervals from $S$ such that consecutive pairs of points in that block are discriminated, and (ii) for each free region, choose intervals from $S$ such that all its points are covered, and the pairs of consecutive points are discriminated.

	\begin{observation}\label{obsv123}
		There exists no interval $I \in S$ that contains both a point in $F_i$ and a point in $F_{i+1}$.
	\end{observation} 
	\begin{proof}
		Note that, $F_i$ and $F_{i+1}$ are sepatated by the block $B_{i+1}$. If there exists an interval $I$ that contains a point in $F_i$ and a point in $F_{i+1}$, then $I$ will contain the point $q_j\in B_{i+1}$ just to the right of $F_i$, which is the reference point of the leftmost group-range $G_j$ of the block $B_{i+1}$. This contradicts the existence of $I \in S$. Also, the size of $I$ then has to be greater than one, which is impossible. \hfill \qed 
	\end{proof}
	Thus, the discriminating code for a free region $F_i$ is disjoint from that of its neighboring free region $F_{i+1}$. So, we can process the free regions independently. 

	{\bf Processing of a free region:} Let the neighboring group-ranges of a free region $F_i$ be $G_a$ and $G_{a+1}$, respectively. There are at most $\frac{4}{\epsilon}$ points lying between the reference points of $G_a$ and $G_{a+1}$. Among these, several points of $P$ to the right (resp. left) of the reference point of $G_a$ (resp. $G_{a+1}$) are inside \emph{block} $B_i$ (resp. $B_{i+1}$). Thus, there are at most $\frac{4}{\epsilon}$ points in $F_i$. Let $S_{F_i} \subseteq S$ be a set whose intervals cover at least one point of $F_i$. Note that, though we have deleted all the redundant intervals of $S$, there may exist several intervals in $S$  whose one endpoint lies in a gap inside that free region, and their other endpoint lies in distinct gaps of the neighboring block. In Figure~\ref{ptasfig1}, there are some blue intervals which are redundant with respect to the points $F_i\cap P$, but are non-redundant with respect to the whole point set $P$. However, the number of such intervals is at most $\frac{4}{\epsilon}$ due to the definition of $(I_i^1, I_i^2)$ of the right-most group-range of the block $B_i$ and left-most group-range of the block $B_{i+1}$.  

	\begin{figure}[!ht]
		\begin{center}
			\includegraphics[scale=0.65]{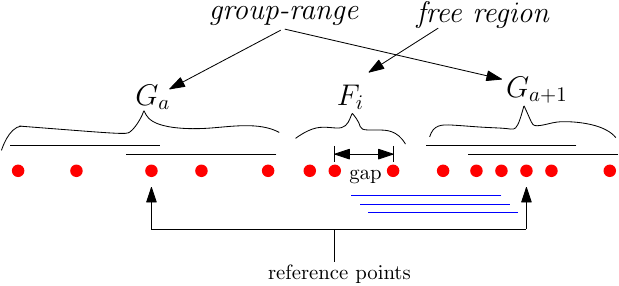}
			\caption{Demonstration of redundant edges in a free region which are non-redundant in the problem instance $(P,S)$.} 
			\label{ptasfig1}
		\end{center}
	\end{figure}

	Thus, we have $|S_{F_i}| = O(1/\epsilon^2)$. We consider all possible subsets of intervals of $S_{F_i}$, and test each of them for being a discriminating code for the points in $F_i$. Let ${\cal D}_i$ be all possible different discriminating codes of the points in $F_i$, with $|{\cal D}_i| =2^{O(1/\epsilon^2)}$ in the worst case. 

	{\bf Processing of a block:} Consider a block $B_i$; its neighboring free regions are $F_i$ and $F_{i+1}$. Consider two discriminating codes $d \in {\cal D}_i$ and $d' \in {\cal D}_{i+1}$. We create a graph $G_i=(V_i,E_i)$ whose nodes $V_i$ correspond to the gaps of $B_i$ which are not discriminated by the intervals used in ${\cal D}_i$ and ${\cal D}_{i+1}$. Each edge $e \in E_i$ corresponds to an interval in $S$ that discriminates pairs of consecutive points corresponding to two different nodes (gaps) of $V_i$. Now, we can discriminate each non-discriminated pair of consecutive points in $B_i$ by computing a minimum edge-cover of $G_i$ in $O(|V_i|^2)$ time~\cite{Vazi}. As mentioned earlier, all the points in $B_i$ are covered. Thus, the discrimination process for the block $B_i$ is over. We will use $\theta(d,d')$ to denote the size of a minimum edge-cover of $B_i$ using $d \in {\cal D}_i$ and $d'\in {\cal D}_{i+1}$.   

	{\bf Computing a discriminating code for $P$:} We now create a multipartite directed graph $H=({\cal D},{\cal F})$. Its $i$-th partite set  corresponds to the discriminating codes in ${\cal D}_i$, and ${\cal D}=\cup_{i=0}^l {\cal D}_i$. Each node $d \in \cal D$ has its weight equal to the size of the discriminating code $d$. A directed edge $(d,d') \in {\cal F}$ connects two nodes $d$ and $d'$ of two adjacent partite sets, say $d \in {\cal D}_i$ and $d' \in {\cal D}_{i+1}$, and has its weight equal to $\theta(d,d')$. For every pair of partite sets ${\cal D}_i$ and ${\cal D}_{i+1}$,  we connect every pair of nodes $(d,d')$, $d \in {\cal D}_i$ and $d' \in {\cal D}_{i+1}$, where $i =0,1,\ldots, l-1$. Every node of ${\cal D}_0$ is connected to a node $s$ with weight~0, and every node of ${\cal D}_l$ is connected to a node $t$ with weight~0.

	\begin{lemma} \label{x}
		The minimum weight $s$-$t$ path in $H$ is a lower bound on the size of the optimum discriminating code for $(P,S)$, where the weight of a path  is equal to the sum of costs of all the vertices and edges on that path.
	\end{lemma}

	\begin{proof} 
		Let $\Pi$ be the minimum weight $s$-$t$ path in the graph $H$, which corresponds to a set of intervals $S'\subseteq S$. To show, $|S'| \leq |S_{opt}|$, where $S_{opt} \subseteq S$ corresponds to the minimum weight discriminating code. For a contradiction, let $|S'| > |S_{opt}|$. As $S_{opt}$ is a discriminating code, the points of every \emph{free region} $F_i$ are discriminated by a subset, say $\delta_i \in S$. Since, we maintain all the discriminating codes in ${\cal D}_i$, surely the subset $\delta_i \in {\cal D}_i$. Let $b_i\subset S$ be the set of intervals that span the points of the block $B_i$. As  $S_{opt}$ is a discriminating code, the points in $B_i$ are discriminated by the intervals in $b_i\cup \delta_i \cup \delta_{i+1}$. Thus, the set of intervals $\beta_i = b_i \setminus (\delta_i \cup \delta_{i+1})$ discriminate the pair of points of $B_i$ that are not discriminated by $\delta_i \cup \delta_{i+1}$. Observe that, for every $i=0,1, \ldots, l$, we have $\delta_i \in {\cal D}_i$. Moreover, there exists a  path $\Pi_{opt}$ that connects $\delta_i, i=0,1,\ldots, l$, whose each edge $(\delta_i,\delta_{i+1})$ has cost equal to $|\beta_i|$. Thus, we have the contradiction that $\Pi_{opt}$ is a path in $H$ having cost less than that of $\Pi$. \hfill \qed
	\end{proof}

	The set of intervals may not form a discriminating code for $P$, as the points in a block may not all be covered. However, the additional intervals $\{(I_i^1,I_i^2), i=1,2, \ldots,\lceil\frac{n\epsilon}{2}\rceil\}$ ensure that the optimum size of the discriminating code satisfies $S_{opt} \geq \lceil\frac{n+1}{2}\rceil$ due to the fact that we have $(n+1)$ gaps, and each interval in $S$ covers exactly 2 gaps. This fact, along with Lemma~\ref{x} implies:

	\begin{lemma}\label{y}
		$|SOL| \leq (1+\epsilon)S_{opt}$.
	\end{lemma}
	\begin{proof} 
		By Lemma~\ref{x}, $|S'| \leq S_{opt}$. The number of extra intervals to cover the \emph{block}s is $\frac{n\epsilon}{2}$. Again, $\frac{n}{2} \leq MEC(P) \leq S_{opt}$, where $MEC(P)$ is the size of minimum edge-cover of the graph $G$ created with the points in $P$ and the intervals in $S$. Thus, $|SOL| \leq (1+\epsilon) S_{opt}$. \hfill \qed
	\end{proof}

	We now analyze the time complexity of the algorithm. Note that, the number of possible discriminating codes in a free region is $2^{O(1/\epsilon^2)}$. Thus, in the graph $H$, the number of edges between a pair of consecutive partite sets ${\cal D}_i$ and ${\cal D}_{i+1}$ is $|{\cal D}_i|\times |{\cal D}_{i+1}| = 2^{O(1/\epsilon^2)}$.  As the computation of the cost of an edge between the sets ${\cal D}_i$ and ${\cal D}_{i+1}$ invokes the edge-cover algorithm of an undirected graph, it needs $O(|B_i|^{2})$ time~\cite{Vazi}. Thus, the total running time of the algorithm is $A+B$, where $A$ is the total time of generating the edge costs, and $B$ is the time for computing a shortest path of $H$. We have $A\leq \sum_{i=1}^{\lceil\frac{n\epsilon}{4}\rceil} 2^{O(1/\epsilon^2)} \times O(|B_i|^2)$. As the $B_i$'s are mutually disjoint, we get $A=O(n^2\times 2^{O(1/\epsilon^2)})$. Moreover, $B=O(|{\cal F}|)=O(\frac{n}{\epsilon}\times 2^{O(1/\epsilon^2)})$~\cite{T99}. 

	In order to reduce the space requirement of the algorithm, we generate partite sets of the multipartite graph $H$ one by one, and compute the length of the shortest path from $s$ up to each node of that set. Initially, the length of the path up to a node $d \in {\cal D}_0$ is $|d|$. While generating ${\cal D}_{i+1}$, the nodes in ${\cal D}_i$ are available along with the length of the shortest path $\chi(d)$ up to each node $d\in {\cal D}_i$ from $s$. Now, we execute the following steps: 
	\begin{description}
		\item[Step 1:] We generate the nodes of ${\cal D}_{i+1}$, and initialize their cost $\chi(.)$ with $\infty$. 
		\item[Step 2:] For each pair of nodes $(d,d'), d \in {\cal D}_i, d' \in {\cal D}_{i+1}$, do the following: 
		\begin{itemize}
			\item Compute the edge cost $\theta(d,d')$, which is the size of the edge-cover of the block $B_i$ using the discriminating codes $d$ of the \emph{free region} $F_i$ and $d'$ of the \emph{free region} $F_{i+1}$. This needs $O(|B_i|^2)$ time using the matching algorithm of an undirected graph~\cite{Vazi}. 
			\item Compute the length of the shortest path from $s$ to $d'$ using the edge $(d,d')$, which is $\chi^*=\chi(d)+\theta(d,d')+|d'|$. 
			\item If the computed length is less than the existing value of $\chi(d')$, then update $\chi(d')$ with $\chi^*$.
		\end{itemize}
	\end{description}

	As the number of discriminating codes in each partite set is $2^{O(1/\epsilon^2)}$ in the worst case which are computed online while considering the $(i+1)$-th partite set, and each discriminating code is of length at most $O(\frac{1}{\epsilon})$, we have the following result.

	\begin{theorem} \label{thm:PTAS}
		The \textsc{Discrete-G-Min-Disc-Code} problem in 1D for unit interval objects admits a PTAS: for every $\epsilon>0$, there is a $(1+\epsilon)$-factor approximation algorithm with time complexity $2^{O(1/\epsilon^2)}n^2$ using $\frac{1}{\epsilon} 2^{O(1/\epsilon^2)}n^2$ space.
	\end{theorem}
	
	Moreover, in this unit interval setting, we easily reduce an instance of \textsc{Continuous-G-Min-Disc-Code} problem to an instance of \textsc{Discrete-G-Min-Disc-Code} problem by first computing the $O(n^2)$ possible non-redundant unit intervals. Thus:
	
	\begin{corollary}\label{cor:PTAS}
		The \textsc{Continuous-G-Min-Disc-Code} problem in 1D for unit interval objects has a PTAS with the same approximation factor, time and space complexity as those for \textsc{Discrete-G-Min-Disc-Code}.
	\end{corollary}

\section{The \textsc{G-Min-Disc-Code} problem in 2D} \label{sec:2D}

	Here, the point set $P=\{p_1,p_2, \ldots, p_n\}$ is given in $\mathbb{R}^2$, and the shape of allowed objects used for covering and discriminating the points of $P$ are axis-parallel squares of equal size. We will use the term {\it unit square} to refer to these objects.

\subsection{NP-completeness}\label{sec:NP-c-2D}

	In~\cite{GledelP19}, it has been shown that \textsc{Continuous-G-Min-Disc-Code} for unit disks in 2D is NP-complete. They reduced the \textsc{$P_3$-Partition-Grid} problem, stated below, to \textsc{Continuous-G-Min-Disc-Code} for unit disks in 2D. We will modify their reduction and apply it to \textsc{Continuous-G-Min-Disc-Code} for axis-parallel unit squares in 2D.
	
	A \emph{grid graph} is a graph whose vertices are positioned in $\mathbb{Z}^2$, and a pair of vertices are adjacent if they are at Euclidean distance~1~\cite{P3P}.
	
	\Pb{\textsc{$P_3$-Partition-Grid} \cite{P3P}}
	{A grid graph $G$.}
	{A partition of the vertices of $G$ into disjoint $P_3$-paths, where a $P_3$-path is a path with three vertices.}
	
	\begin{figure}[ht!]
		\begin{center}
			\includegraphics[scale=0.65]{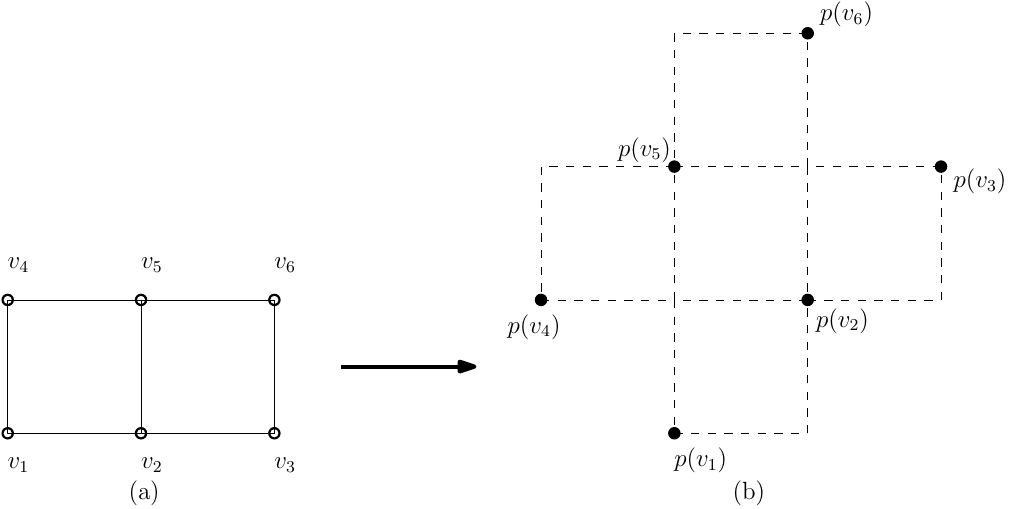}
			\caption{(a) A grid graph G. (b) Its corresponding geometric instance 
			$P_G$, where the dashed axis-parallel unit squares are those covering 
			two points each.} 
			\label{p3}
		\end{center}
	\end{figure}
	
	Given an instance $G$ of \textsc{$P_3$-Partition-Grid}, we construct an instance $P_G$ (a point set) for \textsc{Continuous-G-Min-Disc-Code} as follows. For every vertex $v$ of $G$ with coordinates $(x,y)$, we create a point $p(v)$ with coordinates $(x,y)$ and add it to $P_G$. The construction from~\cite{GledelP19} stops here, and we will now slightly change it. For each point $p(v)$ with coordinates $(x,y)$, we replace it by a point with coordinates $(y-x,y+x)$, that is, we rotate the whole point set by an angle of $\pi/4$ and stretch it by a factor of $\sqrt{2}$ (See Figure~\ref{p3} for an illustration).


	\begin{lemma}\label{lemlem}
		A $P_3$-partition for $G=(V,E)$ exists if and only if there exists a set of $\frac{2|V|}{3}$ axis-parallel unit squares discriminating the points in $P_G$.
	\end{lemma}	
	\begin{proof}
		The key idea is to notice that any axis-parallel unit square can contain at most two points of $P_G$, and if it contains two, then it contains two points corresponding to vertices of $G$ joined by an edge (the center of the square is then placed at the middle-most position of the line segment joining the two points). Moreover, any two points corresponding to an edge of $G$ can be covered by some axis-parallel unit square in that way. Three points corresponding to the three vertices of a $P_3$-path $v_1v_2v_3$ in $G$ can be discriminated using two unit squares $s$ and $s'$, centered at the mid-points of the two segments joining $(p(v_1),p(v_2))$ and $(p(v_2),p(v_3))$, respectively. Now, $p(v_1)$ is covered by $s$ only, $p(v_3)$ by $s'$ only, and $p(v_2)$ by both. Thus, if a $P_3$-partition of $G$ exists, we have our solution of size $\frac{2|V|}{3}$ to the \textsc{Continuous-G-Min-Disc-Code} problem.
		
		Conversely, assume that we have $\frac{2|V|}{3}$ axis-parallel unit squares that discriminate all points of $P_G$. Recall that every square can cover at most two points. For any square $s$ covering two points $p(v_1)$, $p_(v_2)$, we necessarily have that $v_1v_2$ is an edge in $G$. Moreover, one of $p(v_1)$, $p_(v_2)$ needs to be covered by a second square $s'$ (so that the two points are discriminated). Thus, any solution needs at least $\frac{2|V|}{3}$ squares, and any solution of exactly this size will consist of disjoint sets of three points covered by two squares (one point covered by both squares, and the other two, by one of the squares each). These three points must correspond to three vertices of $G$ forming a $P_3$. Thus, we obtain our $P_3$-partition of $G$, as claimed. \hfill \qed
	\end{proof}
		
	Lemma~\ref{lemlem} leads to the following result:		
	\begin{theorem}\label{unitsquares-NPc}
		\textsc{Continuous-G-Min-Disc-Code} and \textsc{Discrete-G-Min-Disc-Code} for axis-parallel unit squares in 2D are NP-complete.
	\end{theorem}
	\begin{proof}
		The statement follows directly from Lemma~\ref{lemlem} in the case of \textsc{Continuous-G-Min-Disc-Code}. Let $S_G$ contain the set of all axis-parallel unit squares that cover two points of $P_G$. For \textsc{Discrete-G-Min-Disc-Code}, we can simply modify the reduction by creating the instance $(P_G,S_G)$ from $G$. \hfill\qed
	\end{proof}

\subsection{Approximation algorithms}\label{appxAlgo}

	We formulate an approximation algorithm by extending the ideas for the 1D case, described in Section~\ref{sec:2-approx}. We will use the techniques of rounding some suitable Integer Linear Programmes (ILPs). Here, our goal is to choose a set $Q$ of points in $\mathbb{R}^2$ of minimum cardinality such that (i) every point of $P$ is covered by at least one axis-parallel unit square among those centered at the points in $Q$ ({\em covering condition}) and (ii) for every pair of points $p_i,p_j \in P$ ($i\neq j$), there exists at least one square in $Q$ whose boundary intersects the interior of the line segment $[p_i, p_j]$ exactly once ({\em discrimination condition}). 

	\begin{figure}[!ht]
		\centering
		\includegraphics[scale=0.65]{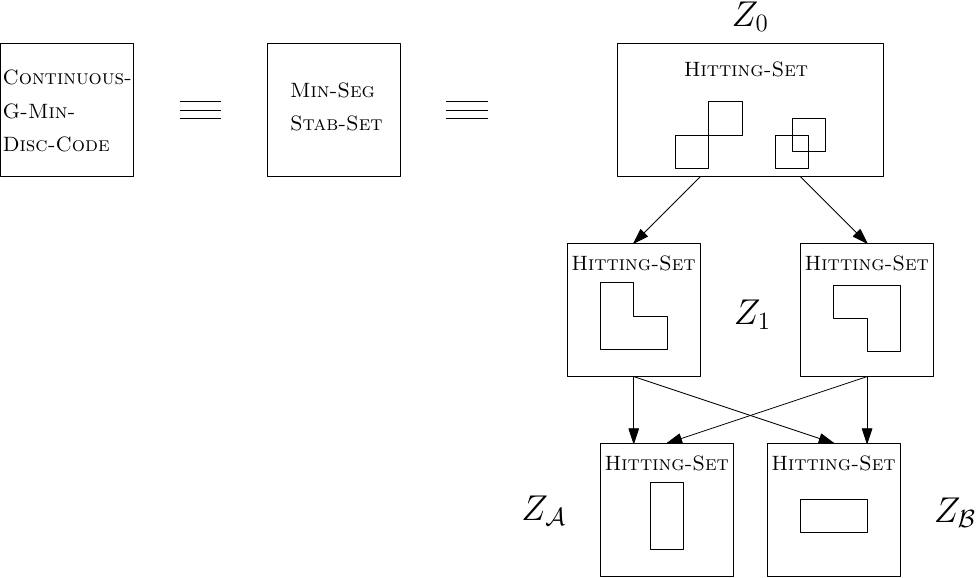}
		\caption{Schematic of the breaking-up of the problem.}
		\label{fig:schema}
	\end{figure}

	We transform our \textsc{Continous-G-Min-Dic-Code} problem into an equivalent problem of segment stabbing (which will be defined below). The segment stabbing problem can also be seen as a hitting set problem of a pair of shapes, called L-shapes. Each of these L-shape hitting set problems is further split into two hitting set problems of unit-height rectangles (or unit-width rectangles). A schematic representation of this process is shown in Figure~\ref{fig:schema}.

	We define the set of line segments $L(P) = \{[p_i, p_j] ~\text{for all}~ p_i,p_j \in P, i\neq j\}$. Thus, the discrimination condition leads to the following problem (where by {\em stabbing} a line segment $\ell$ by a unit square $s$, we mean that exactly one end-point of $\ell$ lies inside $s$).

	\Pb{\textsc{Minimum Segment-Stabbing Set (Min-Seg-Stab-Set)}}
	{A set $L$ of segments in 2D.}
	{A minimum-size set $S$ of axis-parallel unit squares in 2D such that each segment is stabbed by some square of $S$.}

	In fact, \textsc{Min-Seg-Stab-Set} for the input segments $L(P)$ is equivalent to the \textsc{Test Cover} problem for $P$ using axis-parallel unit squares as tests. As in the edge-cover formulation of \textsc{Discrete-G-Min-Disc-Code} problem in 1D (see Section~\ref{sec:2-approx}), here also a feasible solution of \textsc{Min-Seg-Stab-Set} ensures the following:

	\begin{observation} \label{test-C} 
		Every feasible solution $\Phi$ of \textsc{Min-Seg-Stab-Set} (a) discriminates every point-pair in $P$, and (b) at most one point is not covered by any square in $\Phi$.
	\end{observation}    

	In order to discriminate the two endpoints of a member $\ell=[a,b] \in L(P)$, we need to consider the two cases:  $length(\ell) \geq 1$ and  $length(\ell) < 1$, where $length(\ell)$ denotes the length of $\ell$. In the former case, if a center is chosen in any one of the unit squares $D(a)$ and $D(b)$, the segment $\ell$ is stabbed, where $D(q)$ is the axis parallel unit square centered at a point $q$. However, more generally in the second case, to stab $\ell$, we need to choose a center in the region $(D(a)\setminus D(b)) \cup (D(b) \setminus D(a))$. In Figure~\ref{ff2} the shaded region is the feasible region for placing the center of the unit squares to stab a line segment in $L(P)$.   We define a set of distinct objects $\mathcal{O}$ corresponding to the elements of $L(P)$, where each object corresponds to the feasible region of placing the center of a stabbing square of an element of $L(P)$. Thus, the \textsc{Min-Seg-Stab-Set} problem reduces to a \textsc{Hitting-Set} problem, where the objective is to choose a minimum number of points in $\mathbb{R}^2$, such that each object in ${\mathcal O}$ contains at least one of those chosen points. 

	\begin{figure}[!ht]
		\centering
		\includegraphics[scale=0.5]{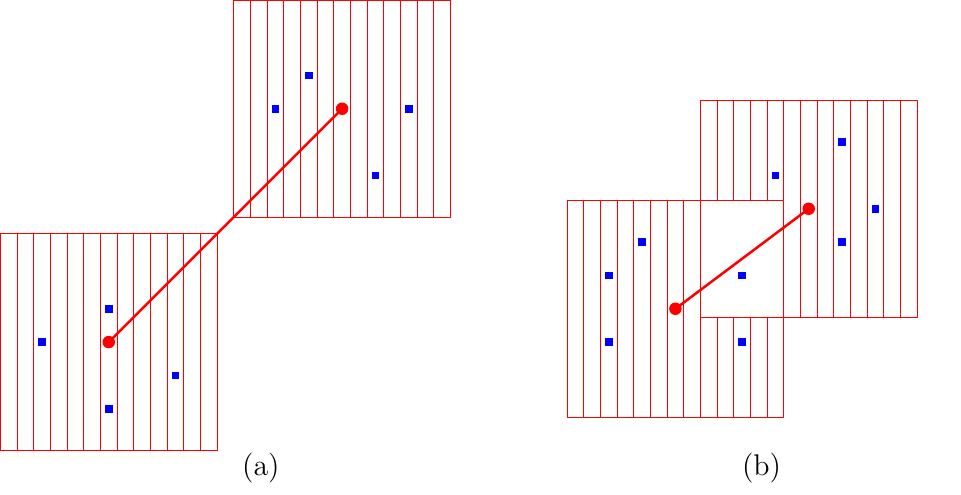}
		\caption{Object that needs to be hit corresponding to segment $\ell=[a,b]$, where (a) $length(\ell) \geq 1$ 		and (b) $length(\ell) < 1$.}
		\label{ff2}
	\end{figure}

	\textbf{The Hitting-Set problem.} We use the technique followed in~\cite{acharya} to solve this problem. Consider the arrangement $\mathcal A$ of the objects in $\mathcal O$. Create a set $Q$ of points by choosing one point in each cell of $\mathcal A$. Thus, the size of the set $Q$ is polynomial in the size of the set $P$. A square centered at a point $q$ inside a cell $A \in {\mathcal A}$ will stab all the segments whose corresponding objects have common intersection region $A$. For each point $q_\alpha \in Q$, we use an indicator variable $x_\alpha$, and can write an integer linear programming (ILP) problem as follows.

	\begin{equation*}\label{eq:ILP}
		\begin{split}
			Z_0: 
			& \min \sum_{\alpha=1}^{|Q|} x_\alpha,\\
			\text{subject to} 
			& ~~\sigma_1(\ell) + \sigma_2(\ell) \geq 1~~ \text{for each segment}~~ \ell=[a,b] \in L(P), \\
			\text{where} 
			& ~~ \sigma_1(\ell)=\sum_{q_\alpha\in Q \cap (D(a)\setminus D(b))} x_\alpha,\\
 			& ~~ \sigma_2(\ell)=\sum_{q_\alpha\in Q \cap (D(b)\setminus D(a))} x_\alpha,\\
			\text{and} 
			& ~~ x_\alpha \in \{0,1\} ~\text{for all points}~ q_\alpha \in Q.
		\end{split}
	\end{equation*}

	As the ILP problem is NP-hard \cite{PS}, we relax the integrality condition of the variables $x_\alpha$ for all $q_\alpha \in Q$ from $Z_0$, and solve the corresponding LP problem in polynomial time. 

	\begin{equation*} \label{eq:LP1}
		\begin{split}
			\overline{Z}_0: 
			& \min \sum_{\alpha=1}^{|Q|} x_\alpha\\ 
			\text{subject to}~ 
			& \sigma_1(\ell) + \sigma_2(\ell) \geq 1 ~\forall~ \ell=[a,b] \in L(P), \\
			~~ \text{and} ~~ 
			& ~~0 \leq x_\alpha \leq 1 ~\forall ~ q_\alpha \in Q.
		\end{split}
	\end{equation*}

	Observe that in the optimum solution $\overline{OPT}_0$ of $\overline{Z}_0$, for each constraint (corresponding to a segment $\ell \in L(P)$), at least one of $\sigma_1(\ell)$ or $\sigma_2(\ell)$ will be greater than or equal to $\frac{1}{2}$. We now define two sets, namely ${\mathcal O}_1$ and ${\mathcal O}_2$. If $\sigma_1(\ell) \geq \frac{1}{2}$ then we put $\ell$ in the set ${\cal O}_1$, and if $\sigma_2(\ell) \geq \frac{1}{2}$ then put $\ell$ in the set ${\mathcal O}_2$. In other words, we choose to hit the objects $(D(p_i)\setminus D(p_j))$ for all $\ell_{ij}=[p_i,p_j]$, $i <j$, if $\sigma_1(\ell_{ij}) \geq \frac{1}{2}$, and choose to hit the objects $(D(p_j) \setminus D(p_i))$  for all $\ell_{ij}=[p_i,p_j]$, $i <j$, if $\sigma_2(\ell_{ij}) \geq \frac{1}{2}$. It needs to be mentioned that, for a constraint corresponding to a point-pair $\ell=[a,b]$ both $\sigma_1(\ell) \geq \frac{1}{2}$ and $\sigma_2(\ell) \geq \frac{1}{2}$ may happen. In that case $\ell$ may be considered in any one the sets ${\cal O}_1$, ${\cal O}_2$ arbitrarily. We form a new ILP as follows: 

	\begin{equation*} \label{eq:LP2}
		\begin{split}
			Z_1: 
			& \min \sum_{\alpha=1}^{|Q|} x_\alpha\\ 
			\text{subject to}~ 
			& \sigma_1(\ell) \geq 1~~ \forall~~ \ell \in {\mathcal O}_1, \\
			& \sigma_2(\ell) \geq 1~~ \forall~~ \ell \in {\mathcal O}_2, \\
			~~ \text{and} ~~ 
			& ~~x_\alpha \in \{0,1\} ~\forall ~ q_\alpha \in Q.
		\end{split}
	\end{equation*}
	
	We use $\overline{Z}_1$ to denote the LP corresponding to the ILP $Z_1$, $OPT_0$ and $OPT_1$, the optimal solutions of $Z_0$ and $Z_1$ respectively,  and $\overline{OPT}_0$ and $\overline{OPT}_1$ the optimal solutions of $\overline{Z}_0$ and $\overline{Z}_1$ respectively. Observe that $2\overline{OPT}_0$ produces a feasible solution to $\overline{Z}_1$. Thus, 
 
	\begin{equation}\label{eq1}
		\overline{OPT}_1 ~~\leq 2\overline{OPT}_0 ~~(\leq 2OPT_0).
	\end{equation}

	However, as the values of the variables in $\overline{OPT}_1$ are fractional, it is not possible to generate a solution of $Z_0$ from $\overline{OPT}_1$. Observe the objects ${\cal O}_1 \cup {\cal O}_2$ considered in $Z_1$ are either a unit square, or a L-shaped object for which one of the length or the width is 1. Thus, our objective is to solve the \textsc{L-Hit} problem, stated below.

		\begin{figure}[!ht]
		\centering
		\includegraphics[scale=0.75]{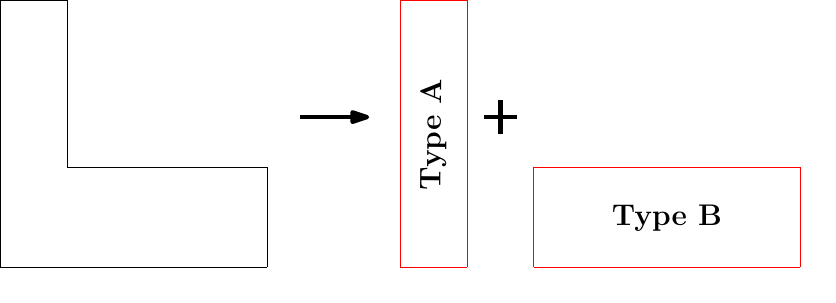}
		\caption{An L-shaped object, which is the union of a type A and a type B object.}
		\label{fzz1}
	\end{figure}

	\textbf{The L-HIT problem.} Here, given a set of L-shaped objects as defined above, and a set of points $Q$, we wish to choose a minimum-size set of points in $Q$ to hit all the L-shaped objects in ${\cal O}_1 \cup {\cal O}_2$. 
	
	We can view an L-shaped object as the union of two rectangles of type $A$ and $B$, where each type $A$ rectangle has height 1 and width less than or {\em equal} to $1$ and each type $B$ rectangle has width 1 and height less than 1 (see Figure~\ref{fzz1}). (Thus, unit squares are considered to be type $A$ rectangles.)
	
	While solving $\overline{Z}_1$, for each constraint (with respect to ${\cal O}_1$ and ${\cal O}_2$) any one or both of the following cases may happen: (a) the
sum of variables whose corresponding points lie in a type $A$ rectangle is $\geq\frac{1}{2}$, and (b) the
sum of variables whose corresponding points lie in a type B rectangle is $\geq\frac{1}{2}$. We accumulate all the rectangles in the set $\cal A$ (resp. $\cal B$) where condition (a) (resp. condition (b)) is satisfied.
The objective is to choose a minimum number of points in $Q$ to hit all the rectangles in ${\cal A}$ and ${\cal B}$. We formulate two ILPs' $Z_{\mathcal A}$ and $Z_{\mathcal B}$ corresponding to two \textsc{Min-UHR-Hit-Set} problems with the set of rectangles $\cal A$ and $\cal B$ respectively, as stated below.
\remove{
Again, as in Equation~\eqref{eq:LP2}, here also we can formulate it as an ILP problem, stated below: 

\begin{equation*} \label{LPAB}
		\begin{split}
			Z_{AB}: 
			& \min \sum_{\alpha=1}^{|Q|} x_\alpha\\ 
			\text{subject to}~ 
			& \sum_{q_\alpha \in A_i\cap Q} x_\alpha \geq 1, \text{for each rectangle} ~ A_i  \in {\mathcal A}, \\, \\
			& \sum_{q_\alpha \in B_i\cap Q} x_\alpha \geq 1, \text{for each rectangle} ~ B_i  \in {\mathcal B}, \\, \\
			~~ \text{and} ~~ 
			& ~~x_\alpha \in \{0,1\} ~\forall ~ q_\alpha \in Q.
		\end{split}
	\end{equation*}
Denoting by $\overline{Z}_{AB}$ the LP version corresponding to $Z_{AB}$, $OPT_{AB}$ and $\overline{OPT}_{AB}$ the solutions of $Z_{AB}$ and $\overline{Z}_{AB}$, we have 

	\begin{equation*}\label{eqAB}
		\overline{OPT}_{AB} \leq 2\overline{OPT}_1 \leq 4\overline{OPT}_0 ~~ \text{(see Equation~\eqref{eq1})}.
	\end{equation*}

Again as mentioned earlier, due to the fractional nature of the solution of $\overline{Z}_{AB}$, it is not possible to generate solution of $Z_{AB}$. Moreover, there may have non-empty intersection among the subset of variables in the solution satisfying constraints corresponding to $\cal A$ and also the constraints corresponding to $\cal B$. Thus, $\overline{OPT}_{AB} \leq \overline{OPT}({\mathcal A})+ \overline{OPT}({\mathcal B})$, where $Z({\mathcal A})$ and $Z({\mathcal B})$ are two ILPs' corresponding to  two \textsc{Min-UHR-Hit-Set} problems with the set of rectangles from $\cal A$, as defined below. (The equivalent problem with rectangles from $\cal B$ can be defined analogosuly.)
}
	
    \Pb{\textsc{Minimum Unit-Height Rectangle Hitting Set (Min-UHR-Hit-Set)}}
	{A set $\cal R$ of unit-height rectangles in $\mathbb{R}^2$.}
	{A set of points that hits all the members of $\cal R$.}

A PTAS for the \textsc{Min-UHR-Hit-Set} problem is known~\cite{Ray-mustafa}; however it cannot be used in Equation~\eqref{eq1} since that does not guarantee any approximation factor for the optimum solution of the corresponding LP problem. However, the \textsc{Min-UHR-Hit-Set} problem for a set of rectangles ${\cal R} = {\cal A}$ (resp. ${\cal B}$) can be formulated as an ILP $Z_{\cal A}$ (resp. $Z_{\cal B}$) as follows:

\begin{minipage}{0.5\textwidth}
		\centering
		\begin{equation*}\label{ILPA}
		\begin{split}
			Z_{\mathcal A}: 
			& \min \sum_{\alpha =1}^{|Q|}x_\alpha,\\
			\text{subject to}\!\!\!
			& \sum_{q_\alpha \in A_i\cap Q} \!\!x_\alpha \geq 1, \forall~ \text{rectangle} ~ A_i  \in {\mathcal A}, \\
			\text{and} ~~
			&~~ x_\alpha \in \{0,1\}, \forall~ \alpha \in Q.
		\end{split}
	\end{equation*} 
	\end{minipage}
	\begin{minipage}{0.5\textwidth}
		\centering
		\begin{equation*}\label{ILPB}
		\begin{split}
			Z_{\mathcal B}: 
			& \min \sum_{\alpha =1}^{|Q|}x_\alpha,\\
			\text{subject to} \!\!\!
			& \sum_{q_\alpha \in B_i\cap Q} \!\!x_\alpha \geq 1, \forall~ \text{rectangle} ~ B_i  \in {\mathcal B}, \\
			\text{and} ~~
			&~~ x_\alpha \in \{0,1\}, \forall~ \alpha \in Q.
		\end{split}
	\end{equation*} 
	\end{minipage}
	
	 Denoting by $\overline{Z}_{\cal A}$ and $\overline{Z}_{\cal B}$ the LP version of $Z_{\cal A}$ and $Z_{\cal B}$, and $\overline{OPT}_{\cal A}$ and  $\overline{OPT}_{\cal B}$ the optimum solutions for $\overline{Z}_{\cal A}$ and $\overline{Z}_{\cal B}$ respectively, we observe that
 $2\overline{OPT}_1$ gives a feasible solution to both $\overline{Z}_{\mathcal A}$ and $\overline{Z}_{\mathcal B}$ simultaneously. Note that, as the variables participating in $\overline{OPT}_1$ and $\overline{OPT}_2$ may not be disjoint, it is not possible to write $\overline{OPT}_{\cal A}+\overline{OPT}_{\cal B}=2\overline{OPT}_1$. However, 
$\overline{Z}_{\mathcal A}\leq 2\overline{OPT}_1$ and $\overline{Z}_{\mathcal B}\leq 2\overline{OPT}_1$.
 Thus by Equation~\eqref{eq1}, we have 
	\begin{equation}\label{eq2}
		\overline{OPT}_{\mathcal A}+\overline{OPT}_{\mathcal B} \leq 4 \overline{OPT}_1 \leq 8\overline{OPT}_0.
	\end{equation}
 
 	For solving the \textsc{Min-UHR-Hit-Set} problem, we use a similar approximation scheme as the one for unit-height rectangles described in \cite{Agarwal}. Here, $\mathbb{R}^2$ is split using $x$-axis parallel lines from a set ${\cal L}=\{\lambda_{1}, \lambda_{2}, \ldots\}$ such that $\lambda_i$ and $\lambda_{i+1}$ are at distance~1 of each other, and such that each rectangle is hit by one of the lines of $\cal L$. Thus, each rectangle in $\cal A$ is intersected by \emph{exactly} one line of $\cal L$ (assuming that no rectangle in  $\cal A$ is aligned with a line in $\cal L$). See Figure~\ref{fig:shiftingStrategy} for a visual representation.
	
	\begin{figure}[!ht]
		\centering
		\includegraphics[scale=0.75]{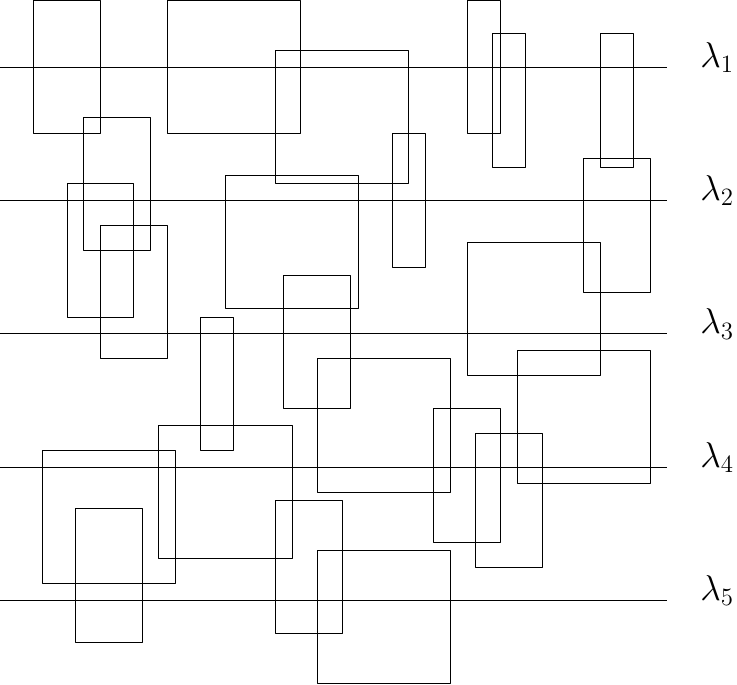}
		\caption{Demonstration of splitting the plane into unit-width horizontal strips; $\lambda_i$s' denote the horizontal lines defining the strips.}
		\label{fig:shiftingStrategy}
	\end{figure}

	Let ${\cal A}_{even}$ (resp.  ${\cal A}_{odd}$) denote the set of rectangles in $\cal A$ that are intersected by  even (resp. odd) numbered lines of $\cal L$. Now, denoting by $Z_{\cal A}$, $Z({\cal A}_{even})$ and $Z({\cal A}_{odd})$ the ILP of the hitting set problems corresponding to the set of rectangles $\cal A$, ${\cal A}_{even}$ and ${\cal A}_{odd}$ respectively, $OPT_{\cal A}$, $OPT({\cal A}_{even})$ and $OPT({\cal A}_{odd})$ as the  optimum solutions of these problems, and $\overline{OPT}_{\cal A}$, $\overline{OPT}({\cal A}_{even})$ and $\overline{OPT}({\cal A}_{odd})$ as the  optimum solutions of the corresponding LP problems, we have
 
	\[\overline{OPT}({\cal A}_{even}) \leq \overline{OPT}_{\cal A} ~~\text{and} ~~\overline{OPT}({\cal A}_{odd}) \leq \overline{OPT}_{\cal A}.\] 

	Now, combining these two inequalities, we have

	\begin{equation} \label{eq3}
		\overline{OPT}({\cal A}_{even}) + \overline{OPT}({\cal A}_{odd}) \leq 2\overline{OPT}_{\cal A}.
	\end{equation}

\green{We will now need the following lemma.}
        
	\begin{lemma}\label{hit-by-horizontal-line}
	\blue{If there exists a horizontal line $\lambda$ that stabs a set of axis-parallel rectangles $\cal D$, then the optimum hitting set for the set $\cal D$ can be computed in polynomial time using the LP relaxation of the corresponding ILP problem.}     
	\end{lemma}
	\begin{proof}
	   \blue{Let $\Pi$ be the optimal hitting set for the rectangles in $\cal D$. Consider the intersection of the elements of $\cal D$ with the line $\lambda$. As $\lambda$ intersects all the members in $\cal D$, if we shift every member of $\Pi$ on $\lambda$ it will not miss to hit any rectangle of $\cal D$ that it was hitting in its earlier position. Thus, our problem reduces to hitting a set of intervals obtained by the intersection of $\lambda$ with the rectangles in $\cal D$.   
	   Considering the arrangement of those intervals and choosing a point in each cell of the arrangement, we can formulate this hitting set problem as an ILP. Again, since the coefficient matrix corresponding to the constraints of this ILP satisfies the consecutive-1-property, we can get the optimal solution of this ILP by solving its LP relaxation~\cite{schrijver2003}.\hfill\qed}
	\end{proof}

Denoting by $Z({\cal A}_i)$ the hitting set problem with the set of  rectangles intersected by $\lambda_i\in \cal L$, $OPT({\cal A}_i)$ and $\overline{OPT}({\cal A}_i)$ as the optimum  solutions for the ILP and LP versions of $Z({\cal A}_i)$, and Lemma~\ref{hit-by-horizontal-line}, we have  $\overline{OPT}({\cal A}_{odd})=OPT({\cal A}_{odd}) = OPT({\cal A}_1) + OPT({\cal A}_3) +  \ldots$. The reason is that for the problem $Z({\cal A}_{odd})$ none of the  rectangles in  ${\cal A}_i$ overlap with any rectangle in  ${\cal A}_j$, for all $i,j$ odd and $i\neq j$. Thus, the hitting set problem for those instances can be solved independently. Similarly, we have $\overline{OPT}({\cal A}_{even})=OPT({\cal A}_{even}) = OPT({\cal A}_2) + OPT({\cal A}_4) + \cdots$. 

	Equations~\eqref{eq2}, \eqref{eq3} and the subsequent discussions lead to the following. Considering the previously computed optimal solutions $SOL(Z_{\cal A})$ and $SOL(Z_{\cal B})$ for $Z_{\cal A}$ and $Z_{\cal B}$, which, by the previous discussions, together form a solution for $Z_0$, we obtain the following chain of inequalities.

	\begin{equation}\label{xyz} \vspace{-0.1in}
		\begin{split}
			|SOL(Z_{\cal A})| + |SOL(Z_{\cal B})|
			& = |SOL({\cal A}_{even})| + |SOL({\cal A}_{odd})| + |SOL({\cal B}_{even})| + |SOL({\cal B}_{odd})|\\
			& =   \overline{OPT}({\cal A}_{even}) + \overline{OPT}({\cal A}_{odd}) + \overline{OPT}({\cal B}_{even}) + \overline{OPT}({\cal B}_{odd}) \text{ \green{by Lemma~\ref{hit-by-horizontal-line}}}\\
			& \leq 2 \times (\overline{OPT}(Z_{\cal A}) + \overline{OPT}(Z_{\cal B})) \text{    by Equation~\eqref{eq3}}\\
			&\leq 16 \times \overline{OPT}_0 \text{    by Equation~\eqref{eq2}}\\
			&\leq 16 \times OPT_0.
		\end{split} \vspace{-0.1in}
	\end{equation}

        \green{Thus, we have proved the following.}
        
	\begin{lemma} \label{lemm}
		The aforesaid algorithm computes a 16-factor approximate solution for \textsc{Min-Seg-Stab-Set}.  
	\end{lemma} 

	We accumulate the hitting set for type $A$ rectangles corresponding to each horizontal line  and the hitting set for type $B$ rectangles corresponding to each vertical line in a set $Q^*$. By Observation~\ref{test-C}, at most one point in $P$ may not be covered by the squares centered at the points of $Q^*$. Thus, we  may require at most one extra square to cover that uncovered point. Thus, we have the following.
		
	\begin{theorem} \label{8approx}
There exists a polynomial-time algorithm for the \textsc{Continuous-G-Min-Disc-Code} problem for axis-parallel unit squares which produces a solution of size at most $16 \cdot OPT+1$, where $OPT$ is the size of an optimal solution. 
 	\end{theorem}

\subsection{Approximation algorithm for \textsc{Discrete-G-Min-Disc-Code}}\label{disscrete}

In this section, we modify the algorithm of Section~\ref{appxAlgo} for \textsc{Continuous-G-Min-Disc-Code} to solve \textsc{Discrete-G-Min-Disc-Code}. Recall that here, in addition to the set of points $P$ (in $\mathbb{R}^2$), the set $S$ of axis-parallel unit squares is also given in the input. As in Section~\ref{appxAlgo}, \textsc{Discrete-G-Min-Disc-Code} reduces to the \emph{discrete} version of the \textsc{Min-UHR-Hit-Set} problem, whose objective is to hit a set $\cal O$ of unit width/height rectangles by choosing a minimum cardinality subset of a given set of points $Q$. Unlike the continuous version (Lemma~\ref{hit-by-horizontal-line}), the discrete version of the hitting set problem for a set of unit-height rectangles intersected by a horizontal line cannot be solved in polynomial time, since the points to be used for hitting the rectangles are already specified (cannot be chosen suitably). However, if we can design an \green{LP-based} $\alpha$-factor approximation algorithm for the discrete version of \textsc{Min-UHR-Hit-Set} when the unit-height rectangles are intersected by a single horizontal line, then we can use it to get a $16\alpha$-factor approximation algorithm for the discrete version of \textsc{Discrete-G-Min-Disc-Code} problem (see Equations~\eqref{eq3} and~\eqref{xyz}). \green{(By \emph{LP-based} we mean that we must be able to compute a solution that is at most $\alpha$ times the optimal solution size of the LP for \textsc{Discrete-Min-UHR-Hit-Set} when all rectangles are intersected by a horizontal line.)}

	\green{We will describe an LP-based $4$-factor approximation algorithm for the \textsc{Discrete-Min-UHR-Hit-Set} problem where the rectangles are hit by a horizontal line (see Lemma~\ref{UHTrestricted} stated at the end of the section). Thus, this will imply the following.}

\begin{theorem} \label{Dapprox}
There exists a polynomial-time algorithm for the \textsc{Discrete-G-Min-Disc-Code} problem for axis-parallel unit squares which produces a solution of size at most $64 \cdot OPT+1$, where $OPT$ is the size of an optimal solution.
\end{theorem}
	
	\begin{figure}[!ht] \vspace{-0.1in}
			\centering
			\includegraphics[scale=0.6]{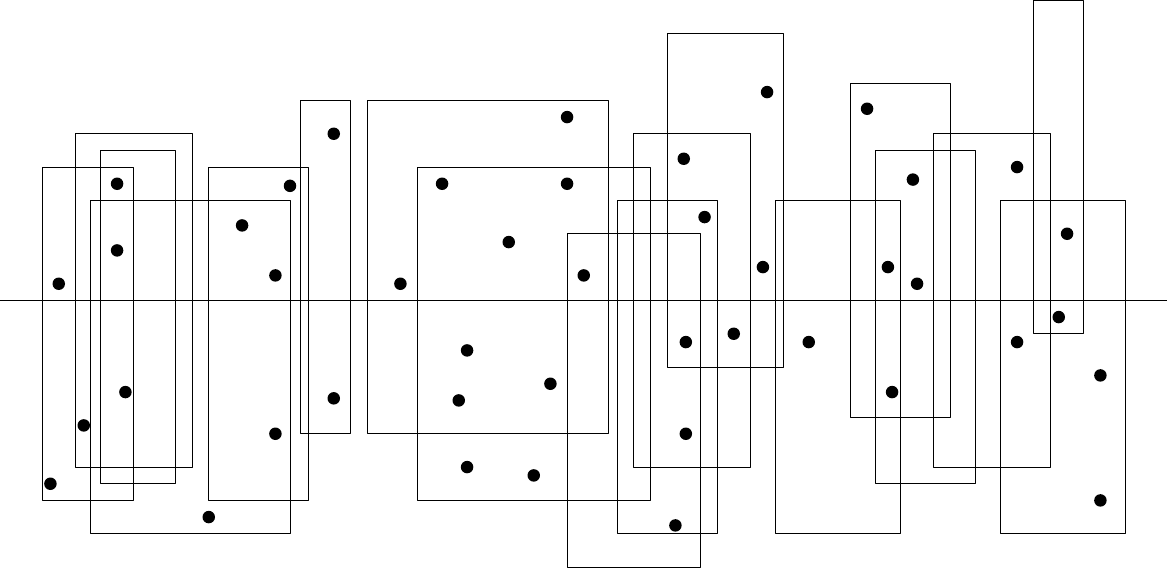}
			\caption{An instance of discrete hitting set of unit-height rectangles stabbed by a horizontal line.}
			\label{discHitHorSeg} \vspace{-0.2in}
		\end{figure}

	\subsubsection{\textsc{Discrete-Min-UHR-Hit-Set} for rectangles stabbed by a single horizontal line}

		Let us first solve a restricted version of the \textsc{Discrete-Min-UHR-Hit-Set} problem, where the input is a set of axis-parallel unit-height rectangles $R$ intersected by a single horizontal line $\lambda$ and a set of points $Q$ (see Figure~\ref{discHitHorSeg}). The objective is to choose a minimum number of points from $Q$ to hit all the rectangles in $R$. This problem can be formulated as the following ILP.
		\begin{equation*} \label{abc}
			\begin{split}
				{\cal U_\lambda}: & \min \sum_{q_\alpha \in Q} x_\alpha\\
				\text{Subject to}~ & \sigma_1(r)+\sigma_2(r) \geq 1, \text{for all} ~r \in R,
			\end{split}
		\end{equation*}
		where $\sigma_1(r)$ (resp. $\sigma_2(r)$) is the sum of the variables corresponding to the points in $Q$ above (resp. below) the line $\lambda$ that lie inside the rectangle $r$. We will use $OPT_{\lambda}$ to denote the optimum solution of this ILP. 

		On the basis of the LP relaxation of this ILP, we can partition the rectangles into two groups: $R^a$ and $R^b$, such that $R^a$ (resp. $R^b$) contains the rectangles whose solution in the LP relaxation satisfies $\sigma_1(r)\geq \sigma_2(r)$ (resp. $\sigma_1(r)<\sigma_2(r)$). The rectangles in $R^a$ (resp. $R^b$) will thus be assumed to be hit by the points in $Q^a$ (resp. $Q^b$) that lie above (resp. below) the line $\lambda$; $Q^a \cup Q^b = Q$, $Q^a \cap Q^b=\emptyset$.
                
		Let ${\cal U}_\lambda^a$ and ${\cal U}_\lambda^b$ be the ILPs for the minimum hitting set problems for the rectangles in $R^a$ and points in $Q^a$ (resp. $R^b$ and $Q^b$). As in the relation of $\overline{OPT}_{\cal A}$, $\overline{OPT}_{\cal B}$ and $\overline{OPT}_1$ in Equation~\eqref{eq2}, here also  denoting by $\overline{OPT}_\lambda^a$ and $\overline{OPT}_\lambda^b$ as the optimum solutions of the LP relaxation of ${\cal U}_\lambda^a$ and ${\cal U}_\lambda^b$, respectively, we have 

		\begin{equation} \label{x1}
			\overline{OPT}_\lambda^a + \overline{OPT}_\lambda^b \leq 2\overline{OPT}_\lambda
		\end{equation}

		due to the fact that $Q^a$ and $Q^b $ are disjoint point sets and $2\overline{OPT}_\lambda$ is a feasible solution of both the problems $\overline{\cal U}_\lambda^a$ and $\overline{\cal U}_\lambda^b$. Now, we concentrate on solving the ILP ${\cal U}_\lambda^a$. The problem ${\cal U}_\lambda^b$ can be solved in a similar manner.

	\subsubsection{Approximation algorithm for solving ${\cal U}_\lambda^a$}
	
	\begin{figure}[t]
		\centering
		\includegraphics[scale=0.6]{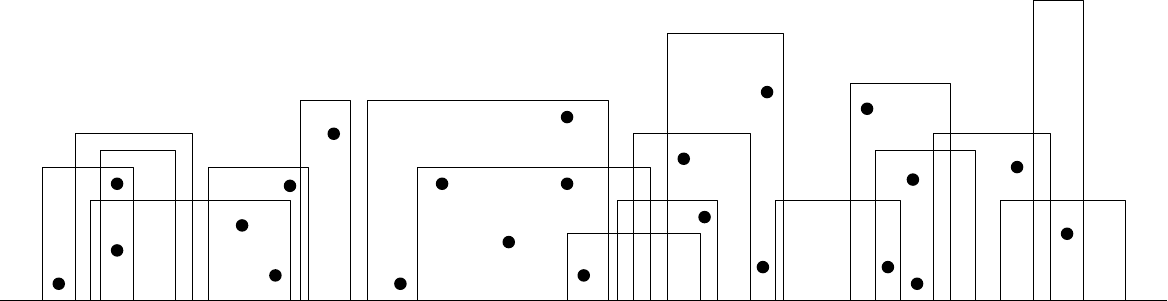}
		\caption{The instance where the rectangles above the horizontal line are considered.}
		\label{discHitHorSegAbove} \vspace{-0.2in}
	\end{figure}

	We have a set of rectangles $R^a$, each one having its bottom boundary aligned with a horizontal line $\lambda$, and a set of points $Q^a\subseteq Q$ that lie above the line $\lambda$. The objective is to find a subset of $Q^a$ of minimum size, say $OPT^a$,  to hit all the members in $R^a$ (see Figure~\ref{discHitHorSegAbove}).

        \green{Note that this problem can be solved in polynomial time by a dynamic programming technique similar to the one used in~\cite{KATZ}. However, in order to plug-in the solution of this problem into Equation~\eqref{x1}, we need that the solution is obtained using the LP relaxation of its corresponding ILP formulation (it may not be optimal, but we need a guaranteed approximation factor).}
        
		\green{We will use known techniques from the literature to get an integer-valued 2-factor approximation algorithm for the ILP problem ${\cal U}_\lambda^a$ obtained from its LP solution (see Equation~\eqref{ua}).}
		
        \begin{equation} \label{ua}
			\begin{split}
				{\cal U}_\lambda^a: & \min \sum_{q_\alpha \in Q^a} x_\alpha\\
				\text{Subject to}~ & \sum_{q_\alpha \in Q^a \cap r} x_\alpha \geq 1  ~\text{for all} ~r \in R^a,\\
				& x_\alpha \in \{0,1\} ~\text{for all}~ q_\alpha \in Q^a 
			\end{split}
		\end{equation}


\begin{definition}\green{Let $\epsilon>0$ be fixed and consider a set system $(X,\mathcal R)$. A set $N\subseteq X$ is an \emph{$\epsilon$-net of $(X,\mathcal R)$} if for every subset $S\in\mathcal R$ for which $|S|\geq\epsilon|X|$, $N$ contains a point of $S$.}
  \end{definition}

\green{  In other words, an $\epsilon$-net is a hitting set of the set system over $X$ containing only large enough sets of $\mathcal R$.}

\begin{lemma}[\cite{Raman}]\label{net-eps}
\green{For a set system $(Q^a,R^a)$ and a horizontal line $\lambda$ such that $Q^a$ is a set of points lying above $\lambda$ and $R^a$ is a set of rectangles each having their base on $\lambda$, for every positive $\epsilon>0$, an $\epsilon$-net of size $\frac{2}{\epsilon}$ can be constructed in polynomial time.} 
\end{lemma}
\begin{proof}
 \green{We split the half-plane above the line $\lambda$ into 
disjoint vertical strips such that each strip contains 
$\lfloor\frac{\epsilon n}{2}\rfloor$ points (see Figure~\ref{fig:eps-net}). Now, from each strip, we choose the bottom-most points. This is an $\epsilon$-net of size $\frac{2}
{\epsilon}$ for the set 
system $(Q^a, R^a)$ due to the fact that for any rectangle 
containing at least $\epsilon n$ points, 
its horizontal span must contain that of a strip, and hence the 
corresponding rectangle is hit by the point chosen in that strip.} \hfill\qed
\end{proof}

\begin{figure}[!ht]
\centerline{\includegraphics[scale=0.6]{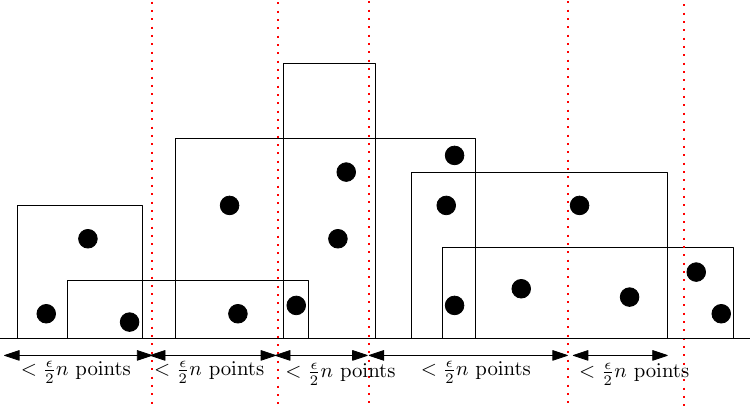}}
\caption{Illustration of the existence of a $\frac{2}{\epsilon}$-size $\epsilon$-net for the set system $(Q^a,R^a)$.}\label{fig:eps-net}
\end{figure}


\green{The following theorem is proved in~\cite{EVEN}.}

		\begin{theorem}[\cite{EVEN}] \label{even}
		  \green{An $\epsilon$-net of size $\frac{d}{\epsilon}$ for a set system $(X, {\cal R})$ with $\epsilon=1/\overline{OPT}$, where $\overline{OPT}$ is the optimal value of the LP for \textsc{Hitting Set} on $(X, {\cal R})$, is a hitting set of $(X, {\cal R})$ of size at most $d\cdot\overline{OPT}$.}
		\end{theorem}
		
		

                \green{                Finally, plugging-in Lemma~\ref{net-eps} and Theorem~\ref{even} with $d=2$ in Equation~\eqref{x1}, we have the following result needed to complete the proof of Theorem~\ref{Dapprox}.
                  }

\begin{lemma} \label{UHTrestricted}
  \green{In polynomial time, one can compute a solution for \textsc{Discrete-Min-UHR-Hit-Set} when all rectangles are intersected by a horizontal line, whose size is at most 4 times the optimal value of the corresponding LP.}
\end{lemma}

\section{\textsc{Min-ID-Code} for geometric intersection graphs} \label{sec:IDcode}
	In this section, we will use techniques similar to those used in the previous sections and apply them to the setting of the graph problem \textsc{Min-ID-Code}, for the intersection graph of axis-parallel unit squares (unit square graphs). To the best of our knowledge, \textsc{Min-ID-Code} was not yet studied for unit square intersection graphs in the literature. 

	Here, the input is a set $S$ of axis-parallel unit squares in 2D. In the graph $G=(V,E)$, the nodes in $V=\{v_1,\ldots,v_n\}$ correspond to the squares in $S$; an edge $e_{ij}=\{v_i,v_j\} \in E$ if the squares corresponding to $v_i,v_j$ intersect. 

	Note that it is not known in the literature whether \textsc{Min-Id-Code} on unit square graphs is NP-hard, however, the techniques from~\cite{MS09} used for unit disk graphs can most certainly be applied to prove it.

	We can reformulate \textsc{Min-ID-Code} for unit square graphs in geometric terms: the objective is to compute a subset $S_{opt} \subseteq S$ of minimum cardinality such that each square in $S$ intersects some square in $S_{opt}$, and for each pair of squares $s_i,s_j\in S$, there exists a square $\sigma \in S_{opt}$ such that $(\sigma \cap s_i \neq \emptyset$ and $\sigma \cap s_j = \emptyset)$ or $(\sigma \cap s_i = \emptyset$ and $\sigma \cap s_j \neq \emptyset)$. If we do only satisfy the discrimination constraint, then in order to satisfy the domination constraint, we may need to include at most one more square from $S$ in $S_{opt}$.

	\begin{figure}[ht!]
		\centering
		\begin{minipage}{0.3\textwidth}
			\centering
			\includegraphics[scale=0.4]{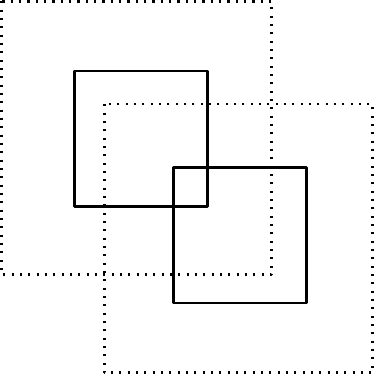} 
		\end{minipage} 
		\begin{minipage}{0.3\textwidth}
			\centering
			\includegraphics[scale=0.4]{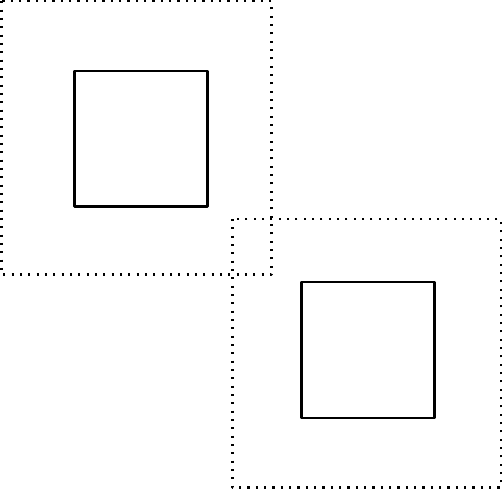} 
		\end{minipage} 
		\begin{minipage}{0.3\textwidth}
			\centering
			\includegraphics[scale=0.4]{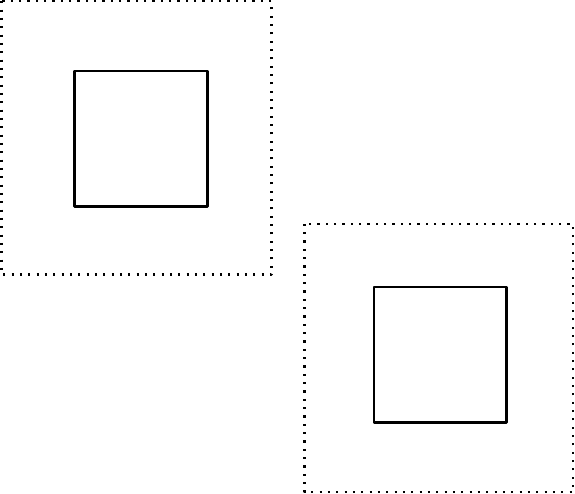} 
		\end{minipage} 
		\caption{Possible intersection patterns of a pair of axis-parallel unit squares (full lines): the dotted square around each square corresponds to the locations where a square centered at this point will intersect the enclosed unit square.}
		\label{figsq}	
	\end{figure}	

	\begin{figure}[ht!]
		\centering
		\begin{minipage}{0.32\textwidth}
			\centering
			\includegraphics[scale=0.4]{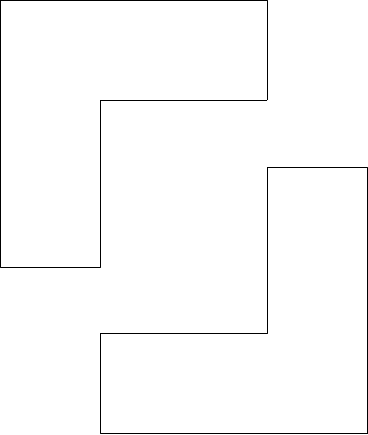} 
		\end{minipage} 
		\begin{minipage}{0.32\textwidth}
			\centering
			\includegraphics[scale=0.4]{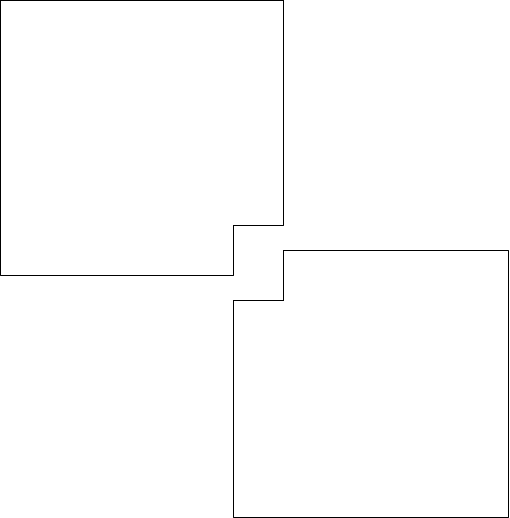} 
		\end{minipage} 
		\begin{minipage}{0.32\textwidth}
			\centering
			\includegraphics[scale=0.4]{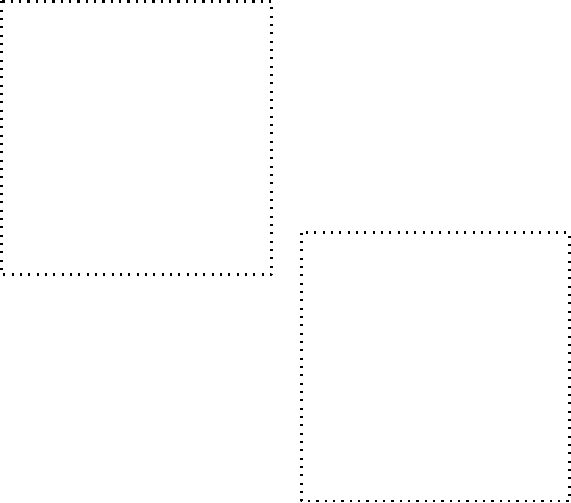} 
		\end{minipage} 
		\caption{Feasible regions for placing the center of the square $s \in ID$ for discriminating $s_i,s_j \in S$: three possible situations.} 
		\label{figsqq}	
	\end{figure}	

	Let $S'\subseteq S$ be an identifying code for the set of squares in $S$. A square $\sigma \in S'$ intersects a square $s_i\in S$ if the center of $\sigma$ is placed inside the square $\delta$ centered at the center of $s$ and the side-length of $\delta$ is twice the side-length of $s$ (shown using dotted line around $s_i$ in Figure~\ref{figsq}). In order to satisfy the discrimination constraint among $s_i,s_j \in S$, the center of a square $\sigma \in S'$ must be placed inside $\delta_i \nabla \delta_j$, where $\nabla$ is the symmetric difference operator, i.e., $(\delta_i \setminus \delta_j) \cup (\delta_j \setminus \delta_i)$. In Figure~\ref{figsq}, different patterns of intersection of a pair $s_i,s_j \in S$ are depicted, along with their covering regions $\delta_i,\delta_j$. Thus, in order to satisfy the discrimination constraint $(s_i,s_j)$, we need to place the center of a square $\sigma \in S'$ in the regions shown in Figure~\ref{figsqq}.

	Thus, as in Section~\ref{appxAlgo}, we can solve \textsc{Min-ID-Code} for unit square graphs by solving a problem of hitting the feasible regions corresponding to each pair of squares $s_i,s_j \in S$ using the centers of the squares in $S$. The objective will be to choose the minimum number of hitting squares from $S$. Thus, the same techniques as in Section~\ref{appxAlgo} can be applied, and we obain the following theorem.

	\begin{theorem} \label{thm:IDcode-4-approx}
		\textsc{Min-ID-Code} has a polynomial-time approximation algorithm for unit square graphs (if the unit square intersection model of the input graph is known) that produces a solution of size at most $64 \cdot OPT+1$, where $OPT$ is the size of an optimal solution.
	\end{theorem}

\section{Conclusion}\label{sec:conclu}
	We have seen that \textsc{Discrete-G-Min-Disc-Code} is NP-complete, even in 1D. This is in contrast with most covering problems and to \textsc{Continuous-G-Min-Disc-Code}, which are polynomial-time solvable in 1D~\cite{covering,GledelP19}.

	We also proposed a simple $2$-factor approximation algorithm for the \textsc{Discrete-G-Min-Disc-Code} problem in 1D, and a PTAS for a special case where each interval in the set $S$ is of unit length. It seems challenging to determine whether \textsc{Discrete-G-Min-Disc-Code} problem in 1D becomes polynomial-time for unit intervals. As noted in~\cite{GledelP19}, this would be related to \textsc{Min-ID-Code} on \emph{unit} interval graphs, which also remains unsolved~\cite{FoucaudMNPV17}. In fact, it also seems to be unknown whether \textsc{Continuous-G-Min-Disc-Code} problem in 1D remains polynomial-time solvable with the restriction that each interval is of unit length. However our PTAS algorithm for \textsc{Discrete-G-Min-Disc-Code} problem in 1D also produces a PTAS for the \textsc{Continuous-G-Min-Disc-Code} problem. We also do not know whether a PTAS exists for the general 1D case.

	In 2D, both \textsc{Continuous-G-Min-Disc-Code} and \textsc{Discrete-G-Min-Disc-Code} problems are NP-complete even when $S$ must be a set of axis-parallel unit square objects. We propose polynomial-time constant factor approximation algorithms for both \textsc{Continuous-G-Min-Disc-Code} and \textsc{Discrete-G-Min-Disc-Code} using the rounding of the relaxation of integer programming to linear programming. The question remains whether there exists any algorithm with better constant approximation or a  PTAS.

        We remark that all our results for axis-parallel unit squares also hold when the objects are axis-parallel rectangles of a specified (same) size i.e., the shape of all the rectangles match in height and width.
	
	As we observed, all the techniques for designing approximation algorithms for solving \textsc{Discrete-G-Min-Disc-Code} problems in 2D work for solving the \textsc{Min-ID-Code} problem for an intersection graph of a set of axis-parallel rectangles of same size. It is an interesting problem whether similar approximation algorithms exist, where the objects in $S$ are unit disks, or arbitrary axis-parallel rectangles.
	
	In the recent literature, a problem known as \textsc{Red-Blue Separation} is being studied. In this problem, given a set $R$ of red-colored points and a set $B$ of blue-colored points in the plane, the objective is to find at most $k$ geometric objects that separate $R$ from $B$, that is, each cell in the arrangement of these geometric objects contains points of at most one color (see~\cite{neeldhara} and the references in that paper). The problem \textsc{G-Min-Disc-Code} studied here can be seen as one where each point has different color. Our techniques for both the continuous and discrete versions of \textsc{G-Min-Disc-Code} problems will also work for the \textsc{Red-Blue-Separation Problem} with axis-parallel rectangles of same size in the continuous and discrete cases, respectively. Here, instead of considering segments joining every pair of points in the given point set, we join each red point in the set $R$ with each point in the set $B$ of blue points, and solve the segment stabbing problem with a set of rectangles (suitably positioned in the continuous case and a given set of rectangles in the discrete case). The approximation factors of both these problems will be the same as those obtained for the corresponding version of \textsc{G-Min-Disc-Code}.

\paragraph{Acknowledgements.} \green{We thank the anonymous referees for their helpful comments that helped us improving the quality of the paper. We also thank Rajiv Raman for helpful discussions and for sharing his lecture notes~\cite{Raman} with us.}

We do not have such Conflict of Interest.
        
\bibliographystyle{abbrv}
\bibliography{bib3}
\end{document}